\documentclass[journal,twocolumn]{IEEEtran}

\usepackage{subcaption}
\usepackage{graphicx}
\usepackage{amsmath}
\usepackage{amssymb }
\usepackage{graphicx}
\usepackage{amsmath}
\usepackage{amssymb}

\usepackage{lipsum}
\usepackage{amsmath}   
\usepackage[noadjust]{cite} 
\usepackage{amssymb}   
\usepackage{amsfonts}   
\usepackage{lipsum}   
\usepackage{float}   
\usepackage{cuted}     
\usepackage{amsthm}
\usepackage[dvipsnames]{xcolor}
\usepackage{hyperref}
\newcommand*{\myprime}{^{\prime}\mkern-1.2mu}
\newcommand*{\mydprime}{^{\prime\prime}\mkern-1.2mu}

\newtheorem{thm}{Theorem}
\newtheorem{lem}[thm]{Lemma}

\newtheorem{cor}{Corollary}

\theoremstyle{definition}

\usepackage{algpseudocode}

\usepackage[ruled,vlined]{algorithm2e}
\def\BibTeX{{\rm B\kern-.05em{\sc i\kern-.025em b}\kern-.08em
    T\kern-.1667em\lower.7ex\hbox{E}\kern-.125emX}}

\ifCLASSINFOpdf
\else
\fi

\begin{document}

\title{Spectral Efficiency Maximization for Active RIS-aided Cell-Free Massive MIMO Systems with Imperfect CSI}

\author{Mahdi Eskandari,
        Huiling Zhu,
        and~Jiangzhou Wang.~\IEEEmembership{Fellow,~IEEE}
        \thanks{Mahdi Eskandari, Huiling Zhu and Jiangzhou Wang are with the School of Engineering, University
of Kent, UK. (email: \{me377, H.Zhu, J.Z.Wang\}@kent.ac.uk). }
        }


\maketitle

\begin{abstract}
A cell-free network merged with active reconfigurable reflecting surfaces (RIS) is investigated in this paper. Considering the presence of imperfect channel state information (CSI), a linear minimum mean square error (LMMSE) technique is initially proposed to estimate the aggregated channel from a user to access points (APs) via RIS. The central processing unit (CPU) then detects uplink data from individual users using the maximum ratio combining (MRC) approach based on the estimated channel. Then, a closed-form expression for uplink spectral efficiency (SE) is derived which demonstrates its reliance on statistical CSI (S-CSI) alone. The amplitude gain of each active RIS element is derived in a closed-form expression as a function of the number of active RIS elements, the number of users, and the size of each reflecting element. A soft actor-critic (SAC) deep reinforcement learning (RL) method is also used to optimize the phase shift of each active RIS element to maximize the uplink SE. Simulation results highlight active RIS' robustness even in the presence of imperfect CSI, proving its efficacy in cell-free networks.
\end{abstract}


\IEEEpeerreviewmaketitle

\ifCLASSOPTIONcaptionsoff
  \newpage
\fi

\section{Introduction}
In recent years, reconfigurable intelligent surfaces (RISs), driven by reflective radios, have attracted significant attention from academia and industry, due to their ability to reconfigure wireless communication environments intelligently and cost-effectively. A RIS consists of multiple programmable elements that can be intelligently adjusted to manipulate the wireless propagation environment, ultimately improving system performance \cite{liu2021reconfigurable, wu2021intelligent}. By strategically designing the reflecting coefficients for each reflecting element (RE), an RIS can effectively improve signal reception at designated destinations and mitigate interference for unintended users ~\cite{di2020smart}. This capability enables the creation of favourable propagation conditions through artificial means and thus enhances existing wireless communication without the need for a more cumbersome RF chain~\cite{renzo2019smart}. 
With RISs, wireless communication systems become more spectrum- and energy-efficient in comparison to conventional systems that rely on complex RF chains.

However, the performance enhancements achievable through the deployment of passive RIS may fall short of expectations in a variety of scenarios, primarily due to the presence of the double path-loss effect over the cascaded channels \cite{long2021active, zhang2022active}. This phenomenon, characterized by a significant attenuation of signal strength over both the direct and reflected paths, can limit the effectiveness of passive RIS in improving communication quality. Consequently, in such cases, alternative strategies or more advanced RIS configurations become necessary to address the challenges posed by the double path-loss effect and achieve the desired level of performance enhancement in wireless communication systems.
An active RIS design was developed in \cite{long2021active} to address the significant double path loss encountered in cascaded channels. Each element of this active RIS architecture is equipped with a reflection-type amplifier, effectively operating as an active reflector. As a result, these amplifiers can enhance electromagnetic (EM) signals. The advantage of active RIS over passive RIS is that it not only adjusts phase shifters to reflect incident signals but also boosts them by using power from an external source.

The active RIS can amplify incoming signals, but it differs significantly in practical terms from amplify-and-forward (AF) relays. The active RIS maintains the fundamental characteristics and hardware structure of conventional passive RIS technology, except for the replacement of reconfigurable passive load impedances with active load impedances. Despite its active load impedances requiring additional power, active RIS's core operation remains the same as passive RIS implementations, involving the adjustment of incident signals directly at the EM level to achieve the desired reflection, as with other passive RIS implementations.
By contrast, the AF relay takes a different approach. In general, RF chains are used for receiving the signal, and then for transmitting it with amplification. To complete the AF process, two time slots are required at the baseband level \cite{long2021active}.

The active RIS retains the benefits of passive RIS technology, with its advantages being showcased in \cite{long2021active}. In \cite{zhang2022active} the sum-rate maximization problem was formulated subject to the joint design of active RIS phase shift and amplitude and the active beamforming at the base station (BS). The downlink energy efficiency (EE) maximization problem of an MU-MISO system was proposed in \cite{ma2022active}, and the active beamforming at the BS and the phase shifts of the active RIS were designed to maximize the EE of the system. In \cite{zhi2022active} the advantages of
active RISs were demonstrated in comparison with passive RISs under the same power budget. 
In \cite{10012424} an active
RIS-aided simultaneous wireless information and power transfer (SWIPT) system
is discussed. The usage of active RIS in integrated sensing and communication (ISAC) in studied in \cite{10054402}. In \cite{10032130} the integration of rate-splitting multiple access (RSMA) and
RIS is considered.
However, the contributions discussed in \cite{long2021active, zhang2022active, ma2022active, zhi2022active, 10012424, 10054402, 10032130} were contingent upon the presumption of having access to instantaneous channel state information (I-CSI), a factor that led to substantial overhead in channel estimation.

To overcome the channel estimation overhead issue, the authors of \cite{peng2022multi} proposed an 
active RIS-aided communication system over spatially correlated
Rician fading channels, where only statistical CSI was available. Also, the work \cite{li2023performance} proposed a linear minimal mean square error (LMMSE) channel estimation method, and the closed-form outage probability of an active RIS system was developed in the presence of phase noise. 

Cell-free massive multi-input multi-output (MIMO) has garnered significant attention in the context of fifth-generation (5G) networks and beyond, as noted in \cite{ngo2017cell}. The integration of cell-free systems and RIS technology has been explored \cite{eskandari2023two, van2021reconfigurable, zhang2021joint, zhang2021beyond}. However, in current research, the impact of channel estimation on active RIS systems is not well investigated and addressed. Therefore, to enhance the capacity of a RIS-assisted cell-free network, this paper introduces the notion of an active RIS-assisted cell-free network. In this network, the traditional RIS is substituted with its active counterpart. The uplink of the cell-free system is assumed in which the transmission is assisted by an active RIS. Under this structure, in the first step, the channel estimation is done using the LMMSE method. The estimated channel is then used to detect the uplink data transmitted by each user with the maximum ratio combining (MRC) approach. In the second step, the spectral efficiency (SE) of each user is calculated at the central processing unit (CPU). Finally, the sum SE maximization problem is formulated concerning the phase shift matrix of the RIS. Finally, the optimization problem is solved using the soft actor-critic (SAC) method \cite{haarnoja2018soft}. The main contribution of the paper is summarized as follows.
\begin{itemize}
    \item The LMMSE channel estimation is proposed for estimating the cascaded channel of a cell-free massive MIMO system assisted by active RIS. In this case, the dimension of the estimated channel is the same as conventional cell-free massive MIMO systems which reduces the pilot overhead significantly.  Furthermore, the closed-form minimum mean squared error is also calculated considering the pilot contamination effect. 
    \item Considering, the practical issues of designing the amplitude of the active RIS, the closed-form expression of the amplitude of the active RIS is derived. Furthermore, the closed form of the SE of each user is derived which is only obtained 
 based on statistical channel state information (S-CSI) and calculated at the CPU. The
impact of the 
channel estimation errors, pilot contamination, spatial
correlation of the channels, and phase shifts of the RIS, which determine
the system performance, are explicitly observable in the derived expression. 
\item  Considering the complexity of the SE maximization problem, the SAC algorithm is adopted to design the phase shifters of the active RIS. The phase shift design of the active RIS is done with S-CSI. Since S-CSI varies much more slowly than the I-CSI, the channel estimation overlead will reduce significantly. SAC is a reinforcement learning framework that merges the actor-critic paradigm with an entropy-regularized approach. Its objective is to optimize not only the expected reward but also the entropy of the policy. This makes SAC particularly effective for tasks occurring in environments that are intricate or uncertain. Given the complexity arising from a substantial number of  RIS elements, devising the RIS phase shifts becomes demanding. However, the encouraging results achieved by SAC render it well-matched for addressing this challenge.
\end{itemize}

The rest of this paper is organized as follows. The details of the active RIS-aided system model are presented in Section~\ref{secII}. In Section~\ref{secIII}, the uplink pilot transmission and channel estimation are provided. In Section~\ref{secIV},
the details of the uplink data transmission are presented. In Section~\ref{secV}, the SAC algorithm is proposed to solve the sun spectral efficiency problem. In Section~\ref{secVI}, numerical results are presented for performance evaluation. Finally, Section~\ref{secVII} concludes the paper.

The major notations in this paper are listed as follows:
 Lowercase, boldface lowercase, and boldface uppercase letters, represented respectively as \(x\), \(\mathbf{x}\), and \(\mathbf{X}\), signify scalar, vector, and matrix entities. The notation \(|x|\) denotes the absolute value of \(x\), while \(|\mathbf{x}|\) signifies the vector's element-wise absolute values. The norm of vector \(x\) is expressed as \(\| \mathbf{x} \|\). The symbol \(\mathcal{CN} (\mu, \sigma^2)\) is employed for denoting the complex Gaussian distribution characterized by a mean \(\mu\) and variance \(\sigma^2\). The statistical expectation is represented as \(\mathbb{E}[\cdot]\), where \(x^*\) denotes the conjugate of \(x\). The transpose and conjugate transpose of matrix \(X\) are respectively denoted as \(X^T\) and \(X^H\). The matrix entry corresponding to the coordinates \((x_1, x_2)\) in \(X\) is expressed as \([\mathbf{X}]_{[x_1, x_2]}\). Additionally, \(\mathrm{tr}\{\cdot\}\) denotes the trace of a matrix.

\section{System Model} \label{secII}
\begin{figure}
    \centering
    \includegraphics[scale=0.25]{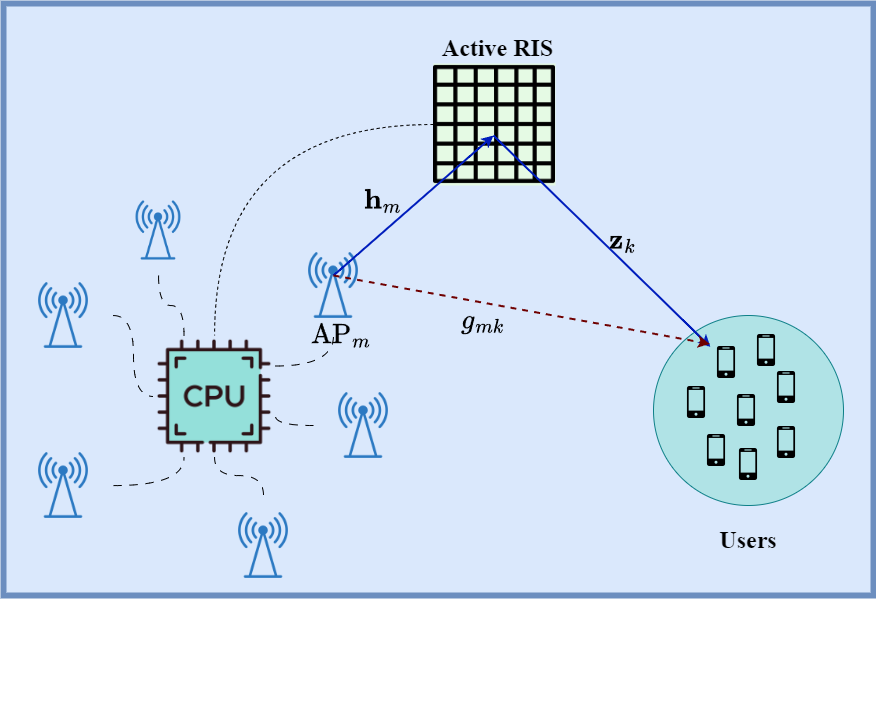}
    \caption{An active RIS-aided Cell-Free Massive MIMO system where $M$ APs collaborate to serve $K$ distant users.}
    \label{fig:sysmodel}
\end{figure}
An active-RIS-aided cell-free massive MIMO system with $M$ single antenna access points (APs) and $K$ single antenna users is considered. The active RIS is equipped with $N$ active elements, each of which is capable of amplifying and reflecting the incident signal. As shown in~Fig.~\ref{fig:sysmodel}, the channel between AP~$m$ and the active RIS, the channel between active RIS and user $k$, and the direct channel from AP $m$ to user $k$ are denoted by $\mathbf{h}_m$, $\mathbf{z}_k$ and $g_{m, k}$, respectively. In particular, the channels are modeled as follows
\begin{align}
    \mathbf{h}_m &\sim \mathcal{CN}(\boldsymbol{0}, \mathbf{R}_m), \\
    \mathbf{z}_k &\sim \mathcal{CN}(\boldsymbol{0},  \Bar{\mathbf{R}}_k), \\
    g_{m, k} &\sim \mathcal{CN}(0, \beta_{m, k}),
\end{align}
where $\beta_{m, k}$ is the large-scale fading coefficient. $ \mathbf{R}_m \in \mathbb{C}^{N \times N}$ and $  \Bar{\mathbf{R}}_k \in \mathbb{C}^{N \times N}$ are the covariance matrices that characterize
the spatial correlation among the channels in $\mathbf{h}_m$ and $ \mathbf{z}_k$, respectively. Specifically, $\mathbf{R}_m$ and $ \Bar{\mathbf{R}}_k$ can be modelled as follows
\begin{align}
    \mathbf{R}_m &= \alpha_m d_H d_V \mathbf{R}, \\
     \Bar{\mathbf{R}}_k &= \bar{\alpha}_k d_H d_V \mathbf{R},
\end{align}
where $\alpha_m$ and $\bar{\alpha}_k$ are the channel large-scale coefficients. $d_H$ and $d_V$ are the horizontal width and the vertical height of each active RIS element, respectively. Also, the $(n_1, n_2)$-th  element of $\mathbf{R} \in \mathbb{C}^{N \times N}$ is $[\mathbf{R}]_{(n_1, n_2)} = \mathrm{sinc}(2\Vert \mathbf{u}_{n_1} - \mathbf{u}_{n_2} \Vert/ \lambda)$ where $\lambda$ is the wavelength and $\mathrm{sinc(x) = \sin{(\pi x)}}/{(\pi x)}$. Moreover, $\mathbf{u}_x = [0, \mod( x-1, N_H )d_H,  \lfloor (x-1)/N_V \rfloor d_V]$ represents the position of element $x$ in $3\mathrm{D}$ space where $N_H$ and $N_V$ denote the total number of RIS elements in each row and
column, respectively. 

For modeling the active RIS, let $\boldsymbol{\Theta} = \mathbf{A} \boldsymbol{\Psi} \in \mathbb{C}^{N \times N}$ denote the reflection matrix of the active RIS, where $\boldsymbol{\Psi} = \mathrm{diag}({\Bar{\psi}}_1, \dots, {\Bar{\psi}}_N)$ represent the reflection coefficients of the active RIS with $ {\Bar{\psi}}_n = e^{j\psi_n}$ being the phase shift of the element $n$. Furthermore, $\mathbf{A} = \mathrm{diag}(a_1, \dots, a_N)$ denote the amplification matrix of the active RIS with $a_n \geq 1$ being the amplification factor of the $n$-th element. Also, different from the
passive RIS, the active RIS incurs non-negligible thermal noise
at all reflecting elements.

According to \cite{long2021active}, the total power consumption of an active RIS system with $N$ elements is given by
\begin{equation}  \label{p_aris}
    P_{\mathrm{ARIS}} = N(P_{\mathrm{c}} + P_{\mathrm{DC}}) + \frac{1}{{\xi}} P_{\mathrm{r}},
\end{equation}
where $\xi$ denotes the amplifier efficiency and  $P_{\mathrm{c}}$ represents the power consumption of the switch and control circuit at each active RIS element; Furthermore, $P_{\mathrm{DC}}$ is the DC biasing power consumption at each active RIS element. Finally $P_{\mathrm{r}} = \sum_{k = 1}^{K} \rho_{u} \mathbb{E}\left\{ \lVert \boldsymbol{\Phi} \mathbf{z}_k \rVert^2 \right\} + \mathbb{E}\left\{ \lVert \boldsymbol{\Phi} \mathbf{v} \rVert^2 \right\}$ is the output power of the active RIS and $\rho_u$ denotes the transmit power of each user. 
\begin{lem}
By assuming all the amplitude gains of the RIS are equal, i.e., $\mathbf{A} = a \mathbf{I}_N$,
the amplitude gain of active RIS is given by
\begin{align} \label{a_nb}
    a = \sqrt{ \frac{\xi(P_{\mathrm{ARIS}}-N(P_{\mathrm{c}} + P_{\mathrm{DC}}))}{
    N \left( \rho_u d_H d_V  \sum_{k=1}^K \Bar{\alpha}_k + \Bar{\sigma}^2
 \right)
    }}.
\end{align}

\begin{proof}
    According to the output power of the active RIS $P_r$, we will have 
\begin{align} \label{amp_diff}
   \sum_{k=1}^K \rho_u  \mathrm{tr}\{ \mathbf{A}^2 \Bar{\mathbf{R}}_k \} + \mathrm{tr}\{ \mathbf{A}^2 \} \Bar{\sigma}^2 = P_r.
\end{align}
    By assuming equal amplitude gain of the RIS, the first term of (\ref{amp_diff}), i.e. $\mathrm{tr}\{ \mathbf{A}^2 \Bar{\mathbf{R}}_k \}$ can be simplified as $\mathrm{tr}\{ \mathbf{A}^2 \Bar{\mathbf{R}}_k \} = a^2 d_H d_V  \bar{\alpha}_k N$. Furthermore the second term of (\ref{amp_diff}) is $\mathrm{tr}\{ \mathbf{A}^2 \} = a^2 N$. Ultimately, through the utilization of the derived expressions and subsequent substitution of (\ref{amp_diff}) into (\ref{p_aris}), coupled with algebraic manipulations, we arrive at the expression delineated in (\ref{a_nb}). 
\end{proof}
Furthermore, the active load of the active RIS can provide a limited amplitude gain $a_{\max}$. Thus, (\ref{a_nb}) can be reformulated as 
\begin{align} \label{a_min}
        a = \min \left\{ \sqrt{ \frac{\xi(P_{\mathrm{ARIS}}-N(P_{\mathrm{c}} + P_{\mathrm{DC}}))}{
    N \left( \rho_u d_H d_V  \sum_{k=1}^K \Bar{\alpha}_k + \Bar{\sigma}^2
 \right)
    }}, a_{\max} \right\}.
\end{align}
\end{lem}

\section{Uplink Pilot Training Phase} \label{secIII}
The channel can be independently estimated by the $\tau_p$ orthogonal pilot sequence sent by all of the $K$ users. It is assumed that the pilot sent by  user~$k$ denoting $\mathbf{s}_k \in \mathbb{C}^{\tau_p \times 1}$ with $\lVert \mathbf{s}_k \rVert^2 = 1$ represents the pilot sequence sent by user~$k$. For the pilot transmission phase, all the $K$ users transmit their pilot sequence to the $M$ APs simultaneously. Specifically, user~$k$ transmits the pilot signal $\sqrt{\rho \tau_p} \mathbf{s}_k$.
Denote the set $\mathcal{P}_k$ as the set of the users that share the same pilot sequence (including user $k$).
Then,  the received signal at AP $m$ denoted as $\mathbf{y}_{m}^{(\mathrm{p})} \in \mathbb{C}^{1 \times \tau_p}$ can be written as 
\begin{align}
    \mathbf{y}_{m}^{(\mathrm{p})} = \sum_{k = 1}^{K} \sqrt{\rho \tau_p} q_{m, k} \mathbf{s}_k^H +  \mathbf{p}_m + \mathbf{w}_{m}^{(\mathrm{p})},
\end{align}
where $q_{m, k} = g_{m, k} + \mathbf{h}_m^H \boldsymbol{\Theta} \mathbf{z}_k$ is the aggregated channel and $\mathbf{p}_m =  \mathbf{h}_m^H \boldsymbol{\Theta} \mathbf{V}$ the equivalent active noise introduced by the active RIS, and $\mathbf{V} \in \mathbb{C}^{N \times \tau_p}$ is the Gaussian additive noise at the active RIS with zero mean and the variance of $\bar{\sigma}^2$. Also, $\mathbf{w}_{m}^{(\mathrm{p})} \sim  \mathcal{CN}(\mathbf{0}, \sigma^2\mathbf{I}_{\tau_p})$ is the additive noise at AP~$m$. Then, AP~$m$ projects the received signal on $\frac{1}{\sqrt{\rho \tau_p}}\mathbf{s}_k$ to estimate the channel from the user $k$ as follows
\begin{align} \nonumber
    y_{m, k}^{(\mathrm{p})} &=   \frac{1}{\sqrt{\rho \tau_p}}\mathbf{y}_{m}^{(\mathrm{p})} \mathbf{s}_k \\ \label{observation_pilot}
    &= q_{m, k} + \sum_{k^{\myprime} \in \mathcal{P}_k \backslash \{ k \}}  q_{m, k^{\myprime}}  + \frac{1}{\sqrt{\rho \tau_p}} \Bar{p}_{m, k}^{(\mathrm{p})} + \frac{1}{\sqrt{\rho \tau_p}} w_{m, k}^{(\mathrm{p})},
\end{align}
where $w_{m, k}^{(\mathrm{p})} = \mathbf{w}_{m}^{(\mathrm{p})}\mathbf{s}_k$ and $\Bar{p}_{m, k}^{(\mathrm{p})} = \mathbf{p}_m \mathbf{s}_k$.

\begin{lem}\label{ch_est}
By assuming that AP~$m$ employs the LMMSE estimation method based on the signal observation
in (\ref{observation_pilot}), the estimate of the aggregated channel $q_{m, k}$ can be formulated as 
\begin{equation}
    \hat{q}_{m, k} = \frac{\mathbb{E} \{ y_{m, k}^{(\mathrm{p})*} q_{m, k} \}}{\mathbb{E} \{ \lVert y_{m, k}^{(\mathrm{p})} \rVert^2 \}} y_{m, k}^{(\mathrm{p})},
\end{equation}
Thus, the  estimated channel could be written as 
\begin{equation}
     \hat{q}_{m, k} = c_{m, k} y_{m, k}^{(\mathrm{p})},
     \label{est_channel}
\end{equation}
with 
\begin{align} \label{c_mk}
    c_{m, k} = \frac{\rho \tau_p \kappa_{m, k} }{\rho \tau_p \sum_{k^{\myprime} \in \mathcal{P}_k} 
 \kappa_{m, k^{\myprime}} + \Bar{\sigma}^2 a^2 \alpha_m N  +  \sigma^2},
\end{align}
where $\kappa_{m, k} \triangleq \beta_{m, k} +  \mathrm{tr}\{ \boldsymbol{\Xi}_{m, k} \}   $ and $\boldsymbol{\Xi}_{m, k} \triangleq a^2 \boldsymbol{\Psi}  \Bar{\mathbf{R}}_k \boldsymbol{\Psi}^H \mathbf{R}_m $. The estimated channel in (\ref{est_channel}) has zero-mean with variance $\gamma_{m, k}$ as follows
\begin{equation}
    \gamma_{m, k} = \mathbb{E} \{ | \hat{q}_{m, k} |^2 \} =  \kappa_{m, k} c_{m, k}.
\end{equation}
Moreover, the estimation error is $e_{m, k} = q_{m, k}  -\hat{q}_{m, k}$ with zero-mean and variance equal to 
\begin{equation}
    \mathbb{E} \{ | e_{m, k} |^2 \} = \kappa_{m, k}  -  \gamma_{m, k}.
\end{equation}
Finally, the normalized mean square error (NMSE) of the channel estimate of the $k$-th at the $m$-AP is given by
\begin{align} \label{NMSE}
    \eta_{m, k} &= \frac{ \mathbb{E} \{ | e_{m, k} |^2 \}}{ \mathbb{E} \{ | q_{m, k} |^2 \}} = 1 - c_{m, k}.
\end{align}
\end{lem}
Based on \cite{van2021outage}, $ \mathrm{NMSE}_{mk}$ is minimized when all the phase shifters of the active RIS are set to be equal. In this case, $\kappa_{m, k} \triangleq \beta_{m, k} + \alpha_m \Bar{\alpha}_k a^2 (d_H d_V)^2 \mathrm{tr}\{ (\mathbf{R})^2 \}   $. 

\section{Uplink Data Transmission and Problem Formulation} \label{secIV}

In the uplink, all the $K$ users transmit their data to all $M$ APs simultaneously. Let the data for user $k$ denote $x_k$ satisfying $\mathbb{E}\{|x_k|^2 \} = 1$. Then, the received signal at AP~$m$ is given by
\begin{align}
    y_{m} = \sqrt{\rho_u} \sum_{k = 1}^{K} q_{m, k} x_k + p_m + w_{m},
\end{align}
where $w_{m} \sim \mathcal{CN}(0, \sigma^2)$ is the additive noise and $p_m =  \mathbf{h}_m^H \boldsymbol{\Theta} \mathbf{v}$. For data detection, the MRC method is applied at the CPU to decode the signal transmitted by each user. Specifically, the decision statistic is
\begin{align} \nonumber
    r_{k} &= \sum_{m = 1}^{M} \hat{q}_{m, k}^* y_{m} \\
    &= \sqrt{\rho_u} \sum_{m = 1}^{M} \sum_{k = 1}^{K} \hat{q}_{m, k}^* q_{m, k} x_k + \hat{q}_{m, k}^* p_m + \hat{q}_{m, k}^* w_{m}.
\end{align}

\begin{lem}
    
The uplink $\mathrm{SINR}$ of user $k$ is given by
\begin{align} \label{sinr_simple}
    \Gamma_{k} = \frac{  I_{1k}^2 } {I_{2k} + I_{3k}}
\end{align}
where
\begin{align} \displaybreak[1]
 I_{1k} &= \sqrt{\rho_u} \sum_{m=1}^M \gamma_{m, k} \\ \nonumber \displaybreak[1]
    I_{2k} &= \rho_u \sum_{k^{\myprime}=1, }^K \sum_{k^{\mydprime} \in \mathcal{P}_k}  \sum_{m = 1}^{M} \sum_{m^{\myprime} = 1}^{M}  c_{m, k} c_{m^{\myprime}, k} \mathrm{tr}\{ \boldsymbol{\Xi}_{m, k^{\myprime}} \boldsymbol{\Xi}_{m^{\myprime}, k^{\mydprime}} \} \\ \nonumber
    &+ \sum_{m = 1}^{M}  \rho_u \gamma_{m, k}^2  + \rho_u  \sum_{k^{\myprime} =1, k^{\myprime} \neq k}^K  
      \sum_{k^{\mydprime} \in \mathcal{P}_k} \sum_{m = 1}^{M} c_{m, k}^2 \kappa_{m, k^{\mydprime}}\kappa_{m, k^{\myprime}} \\ \nonumber \displaybreak[1]
 &+  \frac{\rho_u}{\rho \tau_p} \sum_{k^{\myprime} =1, }^K   \sum_{m = 1}^{M} c_{m, k}^2  \alpha_{m, k^{\myprime}} +  \frac{\rho_u \sigma^2}{\rho \tau_p} \sum_{k^{\myprime} =1, }   \sum_{m = 1}^{M} c_{m, k}^2 \kappa_{m, k^{\myprime}} \\ \nonumber \displaybreak[1]
      &+ \rho_u \sum_{k^{\myprime} \in \mathcal{P}_k \backslash \{k\}}  \left( \sum_{m = 1}^{M}   c_{m, k}  \kappa_{m, k^{\myprime}} \right)^2   \\\label{UI} \displaybreak[1]
      &+ \rho_u \sum_{k^{\myprime} \in \mathcal{P}_k } \sum_{m = 1}^{M}  c_{m, k}^2  \kappa_{m, k^{\myprime}}^2 + \rho_u \sum_{k^{\myprime} \in \mathcal{P}_k }  \sum_{m = 1}^{M}  c_{m, k}^2  \mathrm{tr}\{ \boldsymbol{\Xi}_{m, k^{\myprime}}^2 \} \\
I_{3k} &= \sum_{m=1}^M \alpha_{m, k} +  \sigma^2 \sum_{m = 1}^{M} \kappa_{m, k}.
\end{align}
with
\begin{align}
     \alpha_{m, k} &\triangleq N\bar{\sigma}^2 a^2 \beta_{m, k}\mathrm{tr} \{ \Bar{\mathbf{R}}_m \} + N^2\bar{\sigma}^2 a^4 ( \mathrm{tr}\{\mathbf{R}_m^2 \} \\ \nonumber
     &+ (\mathrm{tr}\{\mathbf{R}_m \})^2 ) \mathrm{tr} \{ \Bar{\mathbf{R}}_k \} 
\end{align}
\end{lem}
\begin{proof}
    See Appendix B.
\end{proof}
The uplink SE is given by 

\begin{align}
    R_k =  \log_2 (1 + \Gamma_{k}).
\end{align}

The optimization problem is formulated as 
maximizing the sum SE of all the users subject to the phase shifts of the active RIS. Specifically, the problem is formulated as follows
\begin{subequations} \label{eq:opt_sum}
\begin{align} \displaybreak[0]
    \max_{\boldsymbol{\Psi}} &\hspace{10mm} \sum_{k=1}^K R_k
    \tag{\ref{eq:opt_sum}} \\
 \text{subject to} \hspace{5mm}
    & |{\Bar{\psi}}_n| = 1, \forall n = 1, \dots, N
\end{align}
\end{subequations}

 Due to the non-convex nature of the optimization problem in (\ref{eq:opt_sum}), and the inability to express it in quadratic form, conventional convex optimization methods are not applicable. In this paper, the use of the soft actor-critic (SAC) algorithm is adopted to solve the problem. Considering SAC offers several advantages for large-scale problems such as handling high-dimensional spaces efficiently, balancing exploration and exploitation effectively, utilizing value function approximation for improved learning, and exhibiting robustness to hyperparameter choices. The details of the SAC algorithm are provided in the next section.

\section{Using SAC to solve (\ref{eq:opt_sum})} \label{secV}
Soft Actor-Critic (SAC) is a powerful reinforcement learning algorithm designed to tackle a wide range of continuous control problems. SAC is rooted in the field of deep reinforcement learning and combines key principles from both actor-critic methods and entropy regularization. Unlike traditional reinforcement learning algorithms that aim solely to maximize cumulative rewards, SAC incorporates the notion of entropy, which measures the randomness or uncertainty of an agent's actions \cite{haarnoja2018soft}. By encouraging higher entropy in the policy, SAC strikes a balance between exploration and exploitation. This means that the agent not only learns to maximize rewards but also maintains a certain level of unpredictability in its actions, allowing it to explore new strategies effectively. SAC achieves this through a combination of a soft value function (hence the name "Soft") and a stochastic policy, making it particularly well-suited for complex and high-dimensional control tasks where exploration is crucial.

One of the primary advantages of SAC lies in its remarkable sample efficiency and stability. It efficiently learns optimal policies with relatively fewer samples compared to other algorithms, making it suitable for real-world applications where collecting data can be costly or time-consuming. Additionally, SAC's ability to handle continuous action spaces and its robustness to hyperparameter choices contribute to its widespread popularity in various domains. Furthermore, the entropy regularization component in SAC leads to more diverse exploration, enabling the algorithm to escape local optima and find better solutions in complex environments.

Consider an unending-horizon discounted Markov decision process (MDP), described by the tuple $(\mathcal{S}, \mathcal{A}, P, r, \gamma)$. Here, $\mathcal{S}$ represents a finite set of states, $\mathcal{A}$ denotes a finite set of available actions, $P: \mathcal{S} \times \mathcal{A} \times \mathcal{S} \rightarrow [0, \infty)$ defines the transition probability distribution for the subsequent state $\mathbf{s}_{t+1} \in \mathcal{S}$ based on the current state $\mathbf{s}_{t} \in \mathcal{S}$ and action $\mathbf{a}_t \in \mathcal{A}$. Additionally, $r: \mathcal{S} \rightarrow \mathbb{R}$ represents the reward function, and $\gamma$ is the discount factor applied to the reward $r$. 
 The objective of RL is to maximize the expected total rewards, which is represented as $\sum_t \mathbb{E}_{(\mathbf{s}_t, \mathbf{a}_t) \sim \rho_\pi}[r(\mathbf{s}_t, \mathbf{a}_t)]$. Here, $\rho_\pi$ is an abbreviation for $\rho_\pi(\mathbf{s}_t, \mathbf{a}_t)$ and signifies the probability distribution of state-action pairs within the trajectory generated by a policy $\pi(\mathbf{a}_t | \mathbf{s}_t)$.
In SAC, the objective function is generalized to the maximum entropy objective which favors stochastic policies by augmenting the objective with the expected entropy of the policy which is
\begin{equation}
    J(\pi) = \sum_{t = 0}^{T} \mathbb{E}_{(\mathbf{s}_t, \mathbf{a}_t) \sim \rho_\pi}[r(\mathbf{s}_t, \mathbf{a}_t) + \nu \mathcal{H}(\pi(. | \mathbf{s}_t))],
\end{equation}
where $T$ is the length of each episode and $   \mathcal{H}(\pi(. | \mathbf{s}_t)) = -\int_{a \in \mathcal{A}} \pi( a | \mathbf{s}_t) \log(\pi( a | \mathbf{s}_t)) da$.
Also, $\nu$ is the temperature parameter and determines the importance of the entropy term against the reward term. Hence, $\nu$ controls the stochasticity of the optimal policy. For the rest of this paper, the temperature explicitly will be omitted, as it can be subsumed into the reward by scaling it by $\nu^{-1}$. Note that maximum entropy RL gradually approaches the conventional RL as $\nu \rightarrow 0$.

SAC employs three distinct neural networks (NNs): a state-value function denoted as $V$ with parameterized weight vectors $\boldsymbol{\Omega}$, a soft Q-function represented as $Q$ with weight parameters $\boldsymbol{\Sigma}$, and a policy function $\pi$ with weight parameters $\boldsymbol{\Pi}$. These function approximations are trained according to the following procedure.

   \subsubsection{Value network} The value function is defined as $ V(\mathbf{s}_t) = \mathbb{E}_{\mathbf{a}_t \sim \pi} \{ Q(\mathbf{s}_t, \mathbf{a}_t) - \log \pi(\mathbf{a}_t \rvert \mathbf{s}_t) \}$
    where $Q(\mathbf{s}_t, \mathbf{a}_t)$ is the Q-function and will be defined in the next part. The value network
    should be trained by minimizing the following error
    \begin{align} \nonumber
        J_V(\boldsymbol{\Omega}) = \mathbb{E}_{\mathbf{s}_t \sim \mathcal{D}} \Big[ \frac{1}{2} \big( V_{\boldsymbol{\Omega}}(\mathbf{s}_t) &- \mathbb{E}_{\mathbf{a}_t \sim \pi_{\boldsymbol{\Pi}}} \big[ Q_{\boldsymbol{\Sigma}}(\mathbf{s}_t, \mathbf{a}_t) \\ \label{value_new}
        &- \log \pi_{\boldsymbol{\Pi}}(\mathbf{a}_t \rvert \mathbf{s}_t)\big]  \big)^2 \Big].
    \end{align}
    The meaning of (\ref{value_new}) is that across all the states that are sampled from the replay buffer $\mathcal{D}$, the value network should be trained with respect to the squared difference between the prediction of the value network and the expected prediction of the Q-function plus the entropy of the policy network.
    Furthermore, the below approximation of the gradient of the $J_V(\boldsymbol{\Omega})$ is used to update the parameters of the $V$ function
    \begin{equation} \label{n_v}
        \hat{\nabla}_{\boldsymbol{\Omega}}J_V(\boldsymbol{\Omega}) =  \nabla_{\boldsymbol{\Omega}}V_{\boldsymbol{\Omega}}(\mathbf{s}_t) (V_{\boldsymbol{\Omega}}\left(\mathbf{s}_t) - Q_{\boldsymbol{\Sigma}} + \log \pi_{\boldsymbol{\Pi}}(\mathbf{a}_t \rvert \mathbf{s}_t) \right).
    \end{equation}
 \subsubsection{Q-network} The Q-network is trained by minimizing the following error  
\begin{equation}
    J_{Q}(\boldsymbol{\Sigma}) = \mathbb{E}_{(\mathbf{s}_t, \mathbf{a}_t) \sim \mathcal{D}} \left[ \frac{1}{2} \left(  Q_{\boldsymbol{\Sigma}}(\mathbf{s}_t, \mathbf{a}_t) - \hat{Q}_{\boldsymbol{\Sigma}}(\mathbf{s}_t, \mathbf{a}_t) \right)^2 \right],
\end{equation}
 where   $ \hat{Q}_{\boldsymbol{\Sigma}}(\mathbf{s}_t, \mathbf{a}_t) = r(\mathbf{s}_t, \mathbf{a}_t) + \gamma \mathbb{E}_{\mathbf{s}_{t+1} \sim p} \left[ V_{\bar{\boldsymbol{\Omega}}}(\mathbf{s}_{t+1}) \right]$,
which can be optimized with a stochastic gradient
\begin{align} \nonumber
    \hat{\nabla}_{\boldsymbol{\Sigma}} J_{Q}(\boldsymbol{\Sigma}) &= \nabla_{\boldsymbol{\Sigma}} Q_{\boldsymbol{\Sigma}}(\mathbf{s}_t, \mathbf{a}_t) \big( Q_{\boldsymbol{\Sigma}}(\mathbf{s}_t, \mathbf{a}_t) \\ \label{grad_q}
    &- r(\mathbf{s}_t, \mathbf{a}_t) - \gamma V_{\bar{\boldsymbol{\Omega}}}(\mathbf{s}_{t+1})  \big).
\end{align}
In (\ref{grad_q}), target value network $V_{\bar{\boldsymbol{\Omega}}}$ is used for the update, where $\bar{\boldsymbol{\Omega}}$ can be an exponentially moving average of the network weights. 
 \subsubsection{Policy network} The policy network is trained by minimizing the following error 
\begin{equation} \label{policy}
    J_{\pi}(\boldsymbol{\Pi}) = \mathbb{E}_{\mathbf{s}_t \sim \mathcal{D}} \left[ D_{\mathrm{KL}} \left( \pi_{\boldsymbol{\Pi}}(. \rvert \mathbf{s}_t) \Bigg{\rVert} \frac{\exp({Q_{\boldsymbol{\Sigma}}(\mathbf{s}_t, .))}}{Z_{\boldsymbol{\Sigma}}(\mathbf
    {s}_t)}\right) \right],
\end{equation}
where $D_{\mathrm{KL}}(P \rVert Q)$ is the Kullback–Leibler (KL) divergence between two probability distributions $P$ and $Q$.
 Note that the KL divergence quantifies how much one probability distribution differs from another probability distribution. If two distributions perfectly match, the KL divergence would be $0$, otherwise it can take values between $0$ and $\infty$.
Hence, the aim of the objective function (\ref{policy}) is to make the distribution of the policy function like the distribution of the exponentiation of the Q-function normalized by another function $Z$. 
In order to minimize the objective function, the authors of \cite{haarnoja2018soft} use a reparameterization trick. In this method, the error back-propagation is ensured by making the sampling process differentiable from the policy. The policy is parameterized as $ \mathbf{a}_t = f_{\mathbf{\Pi}}(\epsilon_t; \mathbf{s}_t)$
where $\epsilon_t$ is a noise vector sampled from a Gaussian distribution.
In this case, the objective function could be written as follows
\begin{align} \nonumber
    J_{\pi}(\boldsymbol{\Pi}) &= \mathbb{E}_{\mathbf{s}_t \sim \mathcal{D}, \epsilon_t \sim \mathcal{N}(0, \bar{\sigma}_\epsilon)} \Big[  \log \pi_{\boldsymbol{\Pi}}(f_{\mathbf{\Pi}}(\epsilon_t; \mathbf{s}_t) \rvert \mathbf{s}_t) \\ \label{trick}
    &- Q_{\boldsymbol{\Sigma}}(\mathbf{s}_t, f_{\mathbf{\Pi}}(\epsilon_t; \mathbf{s}_t)) \Big].
\end{align}
The normalization function $Z$ is removed since it is independent of the parameters $\boldsymbol{\Pi}$.

The unbiased estimator for the gradient of (\ref{trick}) is given by
\begin{align} \nonumber
    \hat{\nabla}_{\boldsymbol{\Xi}} J_{\pi}(\boldsymbol{\Pi}) &= \nabla_{\boldsymbol{\Pi}} \log \pi_{\boldsymbol{\Pi}}(\mathbf{a}_t \rvert \mathbf{s}_t) + \Big( \nabla_{\mathbf{a}_t} \log \pi_{\boldsymbol{\Pi}}(\mathbf{a}_t \rvert \mathbf{s}_t) \\ \label{n_pi}
    &- \nabla_{\mathbf{a}_t} Q(\mathbf{s}_t, \mathbf{a}_t)   \Big) \nabla_{\boldsymbol{\Pi}}f_{\mathbf{\Pi}}(\epsilon_t; \mathbf{s}_t),
\end{align}
where $\mathbf{a}_t$ is evaluated at $f_{\mathbf{\Pi}}(\epsilon_t; \mathbf{s}_t)$.
Note that this algorithm uses two Q-functions to mitigate positive bias in the policy improvement step, which is known to degrade the value-based method's performance. Thus, two Q-functions are parameterized and trained independently, and then the minimum Q-function for the value gradient in (\ref{n_v}) and the policy gradient in (\ref{n_pi}) is used.

 \subsection{ SAC-based Solution}
Generally, SAC is presented as an MDP with observation and action spaces. In the active RIS design problem, the APs, the active RIS, and all the users in the system are denoted by the environment $\mathcal{E}$, while the agent is the CPU that is able to control the active RIS. The following gives the key SAC elements employed to solve the optimization problem (\ref{eq:opt_sum}).
\subsubsection{Observation space}
At each timestep $t$, the observation space is the phases of all the RIS elements $\psi_{n}$. Since $\gamma_{m, k}$ contains the information of both $\kappa_{m, k}$ and $c_{m, k}$, the second part of the conservation space is $\gamma_{m, k}$. Hence, the observation shape is $ N + MK$.
\subsubsection{Action space}
At each timestep, $t$ the action space is the vector containing the phase parts of the phase shifts of the RISs. Thus, the action shape is $N$ and the action range is $[0, 2\pi)$. Noticing $\mathrm{tanh}(.)$ as the activation function for the final layer, which produces the values between $-1$ and $+1$, for converting the result to the desired action range, it is sufficient to set $\mathbf{a}_t = \pi(\mathbf{a}_t^{\prime}+1)$ where $\mathbf{a}_t^{\prime}$ is the output of $\mathrm{tanh}(.)$ activation layer. 
\subsubsection{Reward function} At each timestep $t$, considering the optimization objective in (\ref{eq:opt_sum}), the reward $r_t$ is the sum SE of all the users, i.e.,
$
    r_t = \sum_{k=1}^K R_k(t)
$
where $R_k(t)$ is the SE of user $k$ at time step $t$.

The details of the proposed SAC algorithm for active RIS beamforming are given in \cite{eskandari2023two} and presented in Algorithm~\ref{alg_1}.

\begin{algorithm}
\caption{Soft Actor-Critic}
\label{alg_1}
\SetAlgoLined
 \textbf{Initialization:} Initialize time, states, actions, and replay buffer $\mathcal{D}$ for storing the random states, actions, and rewards in each time step. 
Initialize parameter vectors $\boldsymbol{\Omega}$, $\bar{\boldsymbol{\Omega}}$, $\boldsymbol{\Sigma}$, $\boldsymbol{\Pi}$\;
 \For{\text{each iteration}}{
  Initialize the environment $\mathcal{E}$ and make the initial state $s_0$ \;
  \For{\text{each environment step}}{
  Select the action $\mathbf{a}_t \sim \pi_{\boldsymbol{\Pi}}(\mathbf{s}_t \rvert \mathbf{a}_t)$ \;
  
    Observe the next state $s_{t+1}$ and reward $r_t$, then, store transition set
    $\{\mathbf{s}_t, \mathbf{a}_t, r_t, \mathbf{s}_{t+1} \}$ into $\mathcal{D}$ \;
  
  }
  
  \For{\text{each gradient step}}{
  $\boldsymbol{\Omega} \leftarrow \boldsymbol{\Omega} - \lambda_{V} \hat{\nabla}_{\boldsymbol{\Omega}}J_V(\boldsymbol{\Omega})$ where $\hat{\nabla}_{\boldsymbol{\Omega}}J_V(\boldsymbol{\Omega}) $ is obtained in (\ref{n_v});
  
  $\boldsymbol{\Sigma}_i \leftarrow \boldsymbol{\Sigma}_i - \lambda_{Q} \hat{\nabla}_{\boldsymbol{\Sigma}_i} J_{Q}(\boldsymbol{\Sigma}_i)$ for $i \in \{1, 2 \}$ where $\hat{\nabla}_{\boldsymbol{\Theta}_i} J_{Q}(\boldsymbol{\Theta}_i)$ is obtained in (\ref{grad_q});
  
  $\boldsymbol{\Pi} \leftarrow \boldsymbol{\Pi} - \lambda_{\pi}  \hat{\nabla}_{\boldsymbol{\Pi}} J_{\pi}(\boldsymbol{\Pi})$ where $ \hat{\nabla}_{\boldsymbol{\Pi}} J_{\pi}(\boldsymbol{\Pi})$ is obtained in (\ref{n_pi});
  
  $\bar{\boldsymbol{\Omega}} \leftarrow \bar{\tau}\bar{\boldsymbol{\Omega}} + (1-\bar{\tau}\bar{\boldsymbol{\Omega}})$;
  }
 }
\end{algorithm}

\section{Simulation Results} \label{secVI}
In this section, the effectiveness of the proposed two-timescale architecture is evaluated via simulations. 
A circle-shaped geographic area of radius $500 \mathrm{ m}$ is considered. The APs and the active RIS are located on the diameter of the circle and the users are located inside the circle. The carrier frequency is set to be $1.9 \mathrm{GHz}$. Also, each coherence interval comprises $\tau_c = 200$. For the channels, the large-scale path loss is calculated as $\beta_{m, k} = 10^{-3} (d_{m, k})^{- \bar{\beta} }$, $ \alpha_{m} = 10^{-3} (d_{m})^{- \bar{\alpha}_1 }$, and $\Bar{\alpha}_{k} = 10^{-3} (\Bar{d}_{ k})^{- \bar{\alpha}_2 }$ where $d_{m, k}$ denotes the distance between AP $m$ and user $k$, $d_{m}$ is the distance between the AP~$m$ and the active RIS, and $\Bar{d}_{k}$ represents the distance between user~$k$ and the active RIS, and  Also, the AP-user link path loss exponent, AP-active RIS link path loss exponent, and active RIS-user link path loss exponent are respectively given by 
$ \bar{\beta} = 4$, $ \bar{\alpha}_1 = 2.5 $ and $ \bar{\alpha}_2 = 2.5$. The active and passive noise power is also set to be $\sigma^2 = \Bar{\sigma}^2 = -80 \mathrm{dBm}$. The circuit power consumption is set to be $P_{\mathrm{C}} = -10 \mathrm{dBm}$ and the DC power consumption of each active RIS element is $P_{\mathrm{DC}} = -5 \mathrm{dBm}$. Also, the power consumption of the RIS  is set to be $P_{\mathrm{ARIS}} = 30 \mathrm{dBm}$. The amplifier efficiency is set to be $\xi = 0.8$.  Moreover, the hyperparameters for the SAC algorithm are listed in Table \ref{hp}.

\begin{table}[ht]
\caption{Reinforcement Learning Hyperparameters} 
\centering 
\begin{tabular}{l r } 
\hline\hline 
Parameter & Value \\ [0.5ex] 
\hline 
Learning Rate & $1 \times 10^{-3}$ \\
Discount Factor ($\gamma$) & 0.99 \\
Soft Target Update Rate ($\tau$) & 0.005 \\
Entropy Coefficient ($\alpha$) & 0.2 \\
Batch Size & 64 \\
Replay Buffer Size & 32,000 \\
Neural Network Architecture & Deep neural network \\
Number of Hidden Units & 64 \\
Exploration Noise ($\bar{\sigma}_\epsilon$) & 0.1 \\
Maximum Episode Length & 400 \\
Number of Training Episodes & 2,000 \\
 [1ex]
\hline 
\end{tabular}
\label{hp}
\end{table}

\begin{figure}
    \centering
    \includegraphics[scale = 0.55]{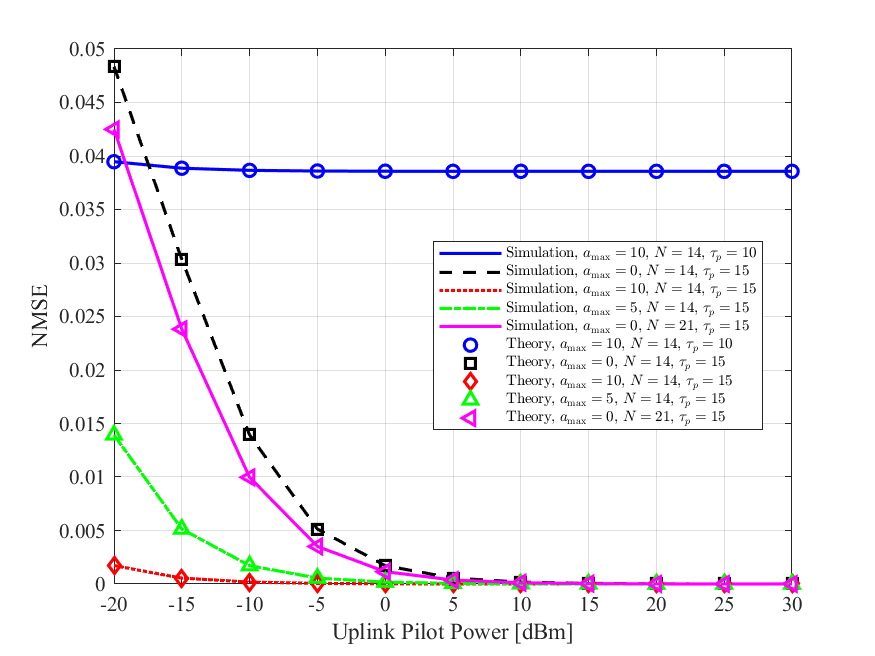}
    \caption{Comparison of the normalized mean square error performance versus
the transmit power~$\rho$ with $K = 15$.}
    \label{fig:nmse}
\end{figure}

Fig~\ref{fig:nmse} shows the NMSE as a function of the uplink power budget of each user. It can be seen from Fig~\ref{fig:nmse} that increasing the amplitude gain of the RIS results in improving the NMSE. Furthermore, increasing the amplitude gain of the RIS is more effective than increasing the number of RIS elements. Also, the effect of pilot contamination could not be canceled by increasing the transmit power. 

\begin{figure}
    \centering
    \includegraphics[scale = 0.55]{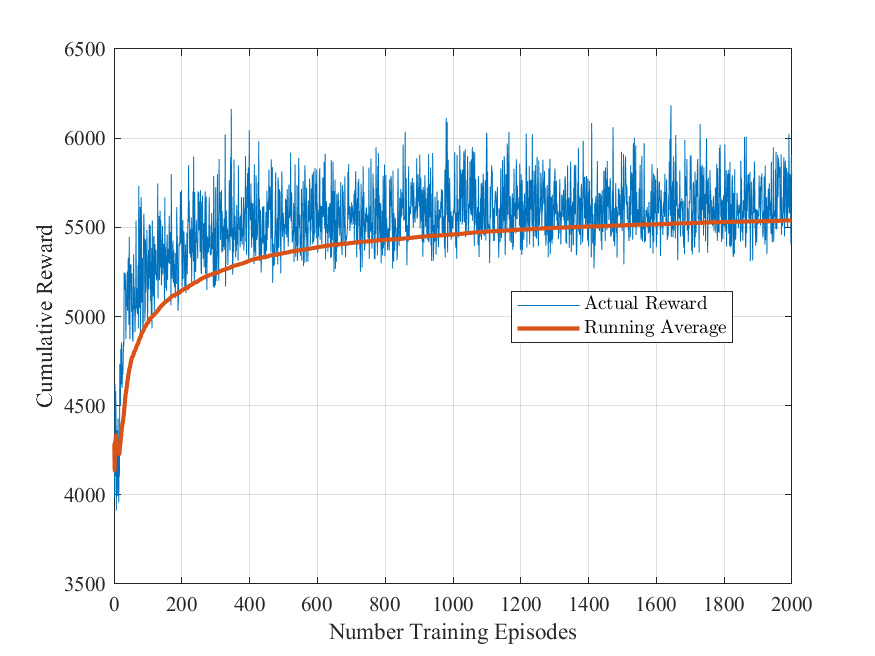}
    \caption{Illustration of the cumulative reward of each episode as a function of training episodes.}
    \label{fig:rew}
\end{figure}

Fig. \ref{fig:rew} shows the cumulative reward of each episode as a function of the number of training episodes. As depicted in the figure, the SAC agent exhibits a significant improvement in cumulative reward as it gains experience through training episodes. The learning curve indicates a clear upward trend, which suggests that the agent is making progressively better decisions in the task of designing RIS phase shifters. Upon reaching approximately $500$ training episodes, a remarkable trend emerges. At this point, the cumulative reward stabilizes and reaches a plateau, indicating that the agent has learned the essential strategies and policies required for efficient RIS phase shifter design.

\begin{figure}
    \centering
    \includegraphics[scale = 0.55]{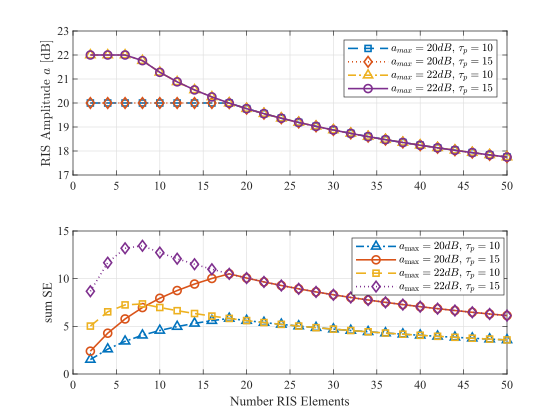}
    \caption{Comparison of the sum SE as a function of the number of RIS elements with $K = 15$.}
    \label{fig:ris}
\end{figure}

Fig.~\ref{fig:ris} depicts the sum SE as a function of the number of RIS elements. It is evident from Fig.~\ref{fig:ris} that when the number of RIS elements increases, initially, the total sum SE also increases. This occurs because, in this scenario, the value of $a \geq a_{\max}$. Consequently, the amplitude gain of the RIS remains fixed at $a_{\max}$.
However, as the number of RIS elements continues to increase, it triggers the $\min$ operator mentioned in  (\ref{a_min}). Since the amplitude gain $a$ in  (\ref{a_nb}) is inversely proportional to $N$, further increments in $N$ result in a reduction in the amplitude gain. Consequently, this leads to a decrease in sum SE as well.

\begin{figure}
    \centering
    \includegraphics[scale = 0.55]{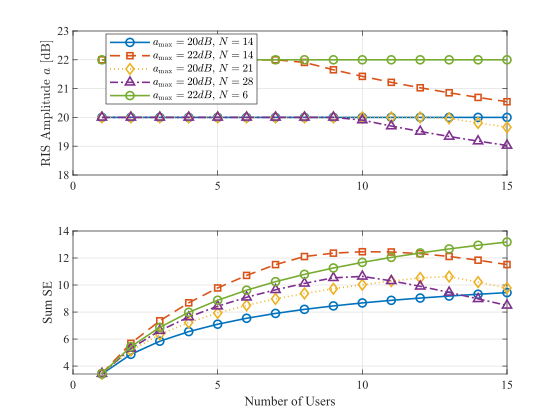}
    \caption{Comparison of the sum SE as a function of the number of users with $M = 20$ and $\tau_p = 15$.}
    \label{fig:user}
\end{figure}

In Fig.~\ref{fig:user} the sum SE as a function of the number of users is presented. It is observed that by increasing the number of users the sum SE is increased, however, by further increasing the number of users the amplitude gain starts to decrease which results in decreasing the sum SE. This is because, similar to the number of RIS elements, the number of users is also inversely proportional to the number of users. 

\begin{figure}
    \centering
    \includegraphics[scale = 0.55]{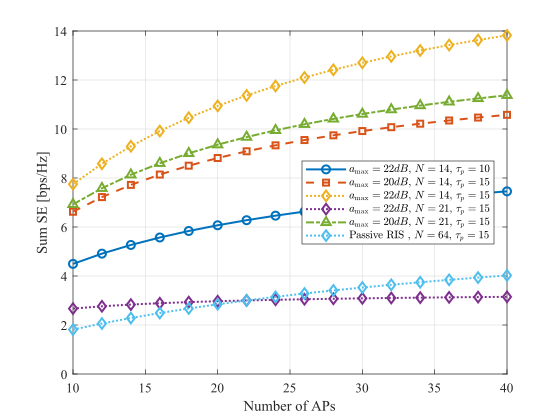}
    \caption{Comparison of sum SE as a function of number of APs with $K = 15$}
    \label{fig:ap}
\end{figure}
Fig.~\ref{fig:ap} shows the sum SE as a function of the number of APs. It can be seen from Fig.~\ref{fig:ap} that by increasing the number of APs, the sum SE also increases. Furthermore, increasing the amplitude gain of the RIS is more effective in improving the sum SE rather than increasing the number of elements. Also, it can be seen that for achieving the same sum SE, less number of APs are required which shows the effectiveness of deploying the RIS panel in a cell-free system. Finally, the effect of pilot contamination could not be mitigated even by increasing the number of APs.  

\begin{figure}
    \centering
    \includegraphics[scale = 0.55]{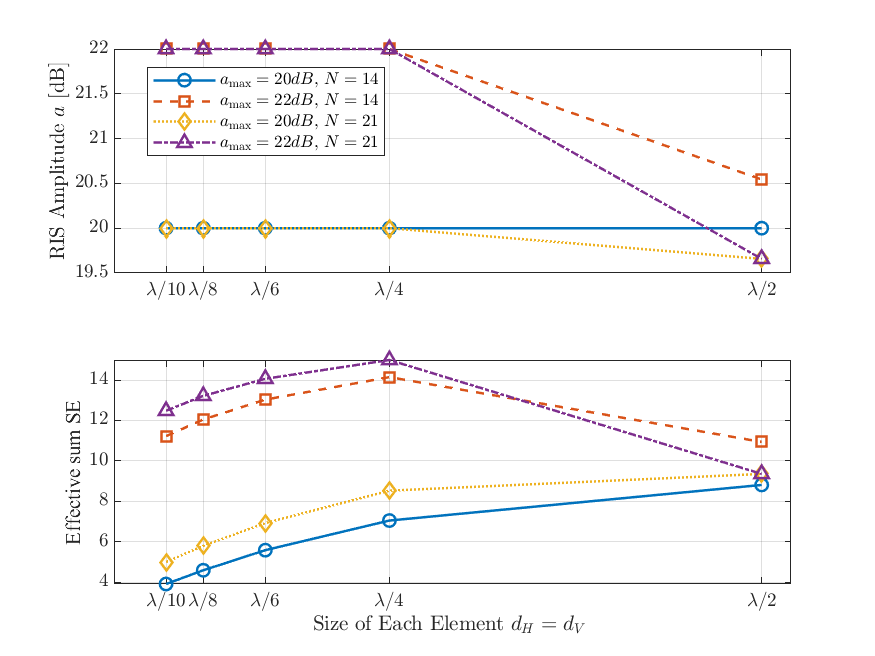}
    \caption{Comparison of the sum SE as a function of the size of each RIS element with  $K = \tau_p = 15$ and $M = 20$.}
    \label{fig:size}
\end{figure}

In Fig.~\ref{fig:size}, the sum SE of the system is presented as a function of the size of each element. It is obvious from Fig.~\ref{fig:size} that by increasing the size of each element, the sum SE also improves. However, by further increasing the size of each element, the amplitude gain of the RIS begins to decrease, the reason is, as shown in (\ref{a_nb}), the size of RIS elements is inversely proportional to the amplitude gain of the RIS and increment of the size of the RIS results in triggering the $\min$ operator in (\ref{a_min}) and the amplitude gain start to decrease. 

\begin{figure}
    \centering
    \includegraphics[scale = 0.55]{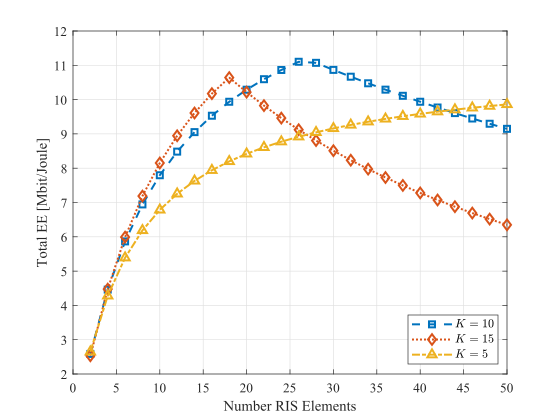}
    \caption{Comparison of sum SE as a function of the number of APs with $a_{\max} = 20 \mathrm{dB}$ and $\tau_p = 15$}
    \label{fig:ee}
\end{figure}

The system's total energy efficiency (EE) comparison is depicted in Fig. \ref{fig:ee} with the total bandwidth of the system being $B = 20 \mathrm{MHz}$.  The proof of EE details can be found in Appendix C. Fig. \ref{fig:ee} indicates that as the number of RIS elements increases, there is an improvement in EE. However, the EE starts to decrease when the amplitude gain of the RIS diminishes due to (\ref{a_min}). Conversely, an initial enhancement in EE is noted with an increasing number of users. Nevertheless, because the number of users is inversely proportional to the amplitude of the RIS, a higher user count leads to a sooner decline in EE.

\section{Conclusion} \label{secVII}
In this paper, the uplink performance of an active RIS-aided cell-free network was explored, accounting for channel estimation inaccuracies. An LMMSE estimator was proposed for estimating the user-AP channels via RIS and derived a closed-form expression for uplink SE. 
The closed-form expression of the amplitude gain of the RIS was derived and its relation with key parameters of the system such as the number of RIS elements is analyzed.  Furthermore, a SAC-based solution 
for designing the phase shift of each element was proposed.  
The following conclusions can be drawn from the results.
\begin{enumerate}
    \item The impact of pilot contamination cannot be ignored even with the high amplitude gain of the active RIS. 
    \item To achieve the same performance, a smaller number of APs is required by using a higher number of RIS elements or higher amplitude gain of the active RIS. Also, in comparison with passive RIS, less number of RIS elements are needed to achieve higher performance gain.
    \item The number of active RIS elements should be designed carefully in each scenario since further increasing the number of reflecting elements results in decreasing the amplitude gain of the RIS.
\end{enumerate}

\section*{Appendix A:  Some Useful Results}
Assume $\mathbf{x} \sim \mathcal{CN}(\mathbf{0}, \mathbf{R}) \in \mathbb{C}^{N \times 1}$, in this case, $\mathbf{W} = \mathbf{x} \mathbf{x}^H$, has a complex central Wishart distribution with 1 degree of freedom and associated
covariance matrix $\mathbf{R}$ which is denoted by $\mathbf{W} \sim \mathcal{CW}_N(\mathbf{R}, 1)$ \cite{tague1994expectations}. 
\begin{lem} {\cite[Eq. (19)]{tague1994expectations}} \label{wishart}
    For $\mathbf{W} \sim \mathcal{CW}_N(\mathbf{R}, 1)$ and a deterministic matrix $\mathbf{A}$, $\mathbb{E}\{\mathbf{W} \mathbf{A} \mathbf{W} \} = \mathbf{R} \mathbf{A} \mathbf{R} + \mathrm{tr}\{ \mathbf{A} \mathbf{R}\}\mathbf{R}$
\end{lem}

\begin{lem}
    The second and the fourth moment of the aggregated channel can be written as
    \begin{align}
        \mathbb{E}\{ |q_{m, k}|^2 \} &= \kappa_{m, k}, \\ \label{e_qq4}
        \mathbb{E}\{ |q_{m, k}|^4 \} &= 2 \kappa_{m, k}^2 + 2 \mathrm{tr}\{ \boldsymbol{\Xi}_{m, k}^2 \}.
    \end{align}
    where $\kappa_{m, k} \triangleq \beta_{m, k} + \mathrm{tr}\{ \boldsymbol{\Xi}_{m, k} \}   $ and $\boldsymbol{\Xi}_{m, k} \triangleq \boldsymbol{\Theta}  \Bar{\mathbf{R}}_k \boldsymbol{\Theta}^H \mathbf{R}_m $. Additionally, the aggregated channels $q_{m, k}$ and $q_{m^{\myprime}k}$ for mutually uncorrelated for $\forall m \neq m^{\myprime}$. Furthermore, aggregated channels $q_{m, k}$ and $q_{m, k^{\myprime}}$ are also mutually uncorrelated for $\forall k \neq k^{\myprime}$. Furthermore, the following conditions hold
    \begin{align}
         &\mathbb{E}\{ |q_{m, k} q_{m^{\myprime}, k^{\myprime}}^*|^2 \} = \kappa_{m, k} \kappa_{m^{\myprime}, k^{\myprime}}\\ \nonumber
         & \hspace{15mm}, \forall m \neq m^{\myprime},  k \neq k^{\myprime}, \\ \label{eqqqq}
         &\mathbb{E}\{ q_{m, k}^* q_{m, k^{\myprime}} q_{m^{\myprime}, k^{\myprime}}^*  q_{m^{\myprime}, k}  \} = \mathrm{tr}\{ \boldsymbol{\Xi}_{m, k^{\myprime}} \boldsymbol{\Xi}_{m^{\myprime}, k} \} \\ \nonumber
         & \hspace{15mm}, \forall m \neq m^{\myprime},  k \neq k^{\myprime}, \\
         &\mathbb{E}\{ |q_{m, k} q_{m^{\myprime}, k}^*|^2 \} =  \kappa_{m, k}\kappa_{m^{\myprime}, k} + \mathrm{tr}\{ \boldsymbol{\Xi}_{m, k} \boldsymbol{\Xi}_{m^{\myprime}, k} \}\\ \nonumber
         & \hspace{15mm}, \forall m \neq m^{\myprime}, \\ \label{E_qmkqmkk}
         &\mathbb{E}\{ |q_{m, k} q_{m, k^{\myprime}}^*|^2 \} =  \kappa_{m, k}\kappa_{m, k^{\myprime}} + \mathrm{tr}\{ \boldsymbol{\Xi}_{m, k} \boldsymbol{\Xi}_{m, k^{\myprime}} \}\\ \nonumber
         & \hspace{15mm}, \forall k \neq k^{\myprime}.
    \end{align}
\end{lem}
\begin{proof}
The proof is similar to the passive RIS scenario in \cite[Lemma 1]{chien2020massive}.
\end{proof}

\begin{lem}
    Assuming $\Bar{p}_{m, k}^{(\mathrm{p})} = \mathbf{p}_m \mathbf{s}_k$ with $\mathbf{p}_m =  \mathbf{h}_m^H \boldsymbol{\Theta} \mathbf{V}$, it holds that 
    \begin{align} \nonumber
      &\mathbb{E} \{ | \Bar{p}_{m, k}^{(\mathrm{p})*}  q_{m, k} |^2 \} = \mathbb{E} \{  | \Bar{p}_{m, k^{\myprime}}^{(\mathrm{p})*}  q_{m, k} |^2 \}  \triangleq \alpha_{m, k} \\ \nonumber
      &= \beta_{m, k} \Bar{\sigma}^2\mathrm{tr}\{ \mathbf{A}^2 \mathbf{R}_m \} \\ \label{e_bq}
      &+ \Bar{\sigma}^2 \mathrm{tr}\{ \boldsymbol{\Theta} \Bar{\mathbf{R}}_k \boldsymbol{\Theta}^H ( \mathbf{R}_m \mathbf{A}^2 \mathbf{R}_m + \mathrm{tr}\{ \mathbf{A}^2 \mathbf{R}_m\}\mathbf{R}_m))  \} \} 
    \end{align}
\end{lem}
\begin{proof}
    \begin{align} \nonumber\displaybreak[1]
         \mathbb{E} \Big\{ | \Bar{p}_{m, k}^{(\mathrm{p})*} & q_{m, k} |^2 \Big\} =  \mathbb{E} \Big\{ \rVert \mathbf{s}_k^H \mathbf{V}^H \boldsymbol{\Theta}^H \mathbf{h}_m ( g_{m, k} + \mathbf{h}_m^H \boldsymbol{\Theta} \mathbf{z}_k) \rVert^2 \Big\}  \\\nonumber \displaybreak[1]
         &=  \mathbb{E} \Big\{ \left(\mathbf{s}_k^H \mathbf{V}^H \boldsymbol{\Theta}^H \mathbf{h}_m ( g_{m, k} + \mathbf{h}_m^H \boldsymbol{\Theta} \mathbf{z}_k) \right)^H \\ \nonumber
         & \hspace{5mm}\times \left(\mathbf{s}_k^H \mathbf{V}^H \boldsymbol{\Theta}^H \mathbf{h}_m ( g_{m, k} + \mathbf{h}_m^H \boldsymbol{\Theta} \mathbf{z}_k) \right) \Big\}  \\
         \label{t_1} \displaybreak[1]
         &= \mathbb{E} \{ \mathbf{s}_k^H \mathbf{V}^H \boldsymbol{\Theta}^H \mathbf{h}_m g_{m, k} g_{m, k}^H \mathbf{h}_m^H \boldsymbol{\Theta}  \mathbf{V}  \mathbf{s}_k\} \\ \label{t_2}  \displaybreak[1]
         & + \mathbb{E} \{ \mathbf{s}_k^H \mathbf{V}^H \boldsymbol{\Theta}^H \mathbf{h}_m g_{m, k} \mathbf{z}_k^H \boldsymbol{\Theta}^H \mathbf{h}_m\mathbf{h}_m^H \boldsymbol{\Theta}  \mathbf{V}  \mathbf{s}_k\} \\  \label{t_3} \displaybreak[1]
         & + \mathbb{E} \{ \mathbf{s}_k^H \mathbf{V}^H \boldsymbol{\Theta}^H \mathbf{h}_m \mathbf{h}_m^H \boldsymbol{\Theta} \mathbf{z}_k g_{m, k}^H \mathbf{h}_m^H \boldsymbol{\Theta}  \mathbf{V}  \mathbf{s}_k\} \\   \label{t_4}\displaybreak[1]
         &+ \mathbb{E} \{ \mathbf{s}_k^H \mathbf{V}^H \boldsymbol{\Theta}^H \mathbf{h}_m \mathbf{h}_m^H \boldsymbol{\Theta} \mathbf{z}_k \mathbf{z}_k^H \boldsymbol{\Theta}^H \mathbf{h}_m \mathbf{h}_m^H \boldsymbol{\Theta}  \mathbf{V}  \mathbf{s}_k\}.
    \end{align} 
    Thanks to the independency between the channels, (\ref{t_2}) and (\ref{t_3}) are equal to zero. Furthermore, (\ref{t_1}) can be written as
    \begin{align} \nonumber
        &\mathbb{E} \{ \mathbf{s}_k^H \mathbf{V}^H \boldsymbol{\Theta}^H \mathbf{h}_m g_{m, k} g_{m, k}^H \mathbf{h}_m^H \boldsymbol{\Theta}  \mathbf{V}  \mathbf{s}_k\}  \\ \nonumber \displaybreak[1]
        &= \mathrm{tr}\{ \mathbb{E} \{ \mathbf{s}_k^H \mathbf{V}^H \boldsymbol{\Theta}^H \mathbf{h}_m g_{m, k} g_{m, k}^H \mathbf{h}_m^H \boldsymbol{\Theta}  \mathbf{V}  \mathbf{s}_k\} \} \\ \nonumber \displaybreak[1]
        &=  \mathrm{tr}\{ \mathbb{E} \{  \boldsymbol{\Theta}^H \mathbf{h}_m g_{m, k} g_{m, k}^H \mathbf{h}_m^H \boldsymbol{\Theta} \} \mathbb{E} \{ \mathbf{V}  \mathbf{s}_k  \mathbf{s}_k^H \mathbf{V}^H   \} \} \\ \nonumber \displaybreak[1]
       &= \Bar{\sigma}^2\mathrm{tr}\{ \mathbb{E} \{  \boldsymbol{\Theta}^H \mathbf{h}_m g_{m, k} g_{m, k}^H \mathbf{h}_m^H \boldsymbol{\Theta} \} \} \\ \nonumber \displaybreak[1]
       &= \Bar{\sigma}^2\mathrm{tr}\{ \mathbb{E} \{    g_{m, k} g_{m, k}^H \}  \mathbb{E} \{ \mathbf{h}_m^H \boldsymbol{\Theta} \boldsymbol{\Theta}^H \mathbf{h}_m \} \} \\ \label{t_1s} \displaybreak[1]
       &= \beta_{m, k} \Bar{\sigma}^2\mathrm{tr}\{ \mathbf{A}^2 \mathbf{R}_m \}.
    \end{align}
    Next, (\ref{t_4}) is given by
    \begin{align} \nonumber
        & \mathbb{E} \{ \mathbf{s}_k^H \mathbf{V}^H \boldsymbol{\Theta}^H \mathbf{h}_m \mathbf{h}_m^H \boldsymbol{\Theta} \mathbf{z}_k \mathbf{z}_k^H \boldsymbol{\Theta}^H \mathbf{h}_m \mathbf{h}_m^H \boldsymbol{\Theta}  \mathbf{V}  \mathbf{s}_k\} \\ \nonumber 
        &= \mathrm{tr}\{ \mathbb{E} \{ \mathbf{s}_k^H \mathbf{V}^H \boldsymbol{\Theta}^H \mathbf{h}_m \mathbf{h}_m^H \boldsymbol{\Theta} \mathbf{z}_k \mathbf{z}_k^H \boldsymbol{\Theta}^H \mathbf{h}_m \mathbf{h}_m^H \boldsymbol{\Theta}  \mathbf{V}  \mathbf{s}_k\}   \} \\ \nonumber
        &=  \mathrm{tr}\{ \mathbb{E} \{  \boldsymbol{\Theta}^H \mathbf{h}_m \mathbf{h}_m^H \boldsymbol{\Theta} \mathbf{z}_k \mathbf{z}_k^H \boldsymbol{\Theta}^H \mathbf{h}_m \mathbf{h}_m^H \boldsymbol{\Theta} \}  \mathbb{E} \{   \mathbf{V}  \mathbf{s}_k \mathbf{s}_k^H \mathbf{V}^H    \}    \} \\ \nonumber
        &= \Bar{\sigma}^2 \mathrm{tr}\{ \mathbb{E} \{    \boldsymbol{\Theta} \mathbf{z}_k \mathbf{z}_k^H \boldsymbol{\Theta}^H \} \mathbb{E} \{ \mathbf{h}_m \mathbf{h}_m^H \mathbf{A}^2 \mathbf{h}_m \mathbf{h}_m^H  \} \} \\ \label{t_4s}
        &= \Bar{\sigma}^2 \mathrm{tr}\{ \boldsymbol{\Theta} \Bar{\mathbf{R}}_k \boldsymbol{\Theta}^H ( \mathbf{R}_m \mathbf{A}^2 \mathbf{R}_m + \mathrm{tr}\{ \mathbf{A}^2 \mathbf{R}_m\}\mathbf{R}_m))  \} \},
    \end{align}
    where the last equality is obtained from Lemma~\ref{wishart}.
    
Finally, by summing up (\ref{t_1s}) and (\ref{t_4s}) we arrive at
\begin{align} \nonumber
      \mathbb{E} \{ | \Bar{p}_{m, k}^{(\mathrm{p})*} & q_{m, k} |^2 \} = \beta_{m, k} \Bar{\sigma}^2\mathrm{tr}\{ \mathbf{A}^2 \mathbf{R}_m \} \\ 
      &+ \Bar{\sigma}^2 \mathrm{tr}\{ \boldsymbol{\Theta} \Bar{\mathbf{R}}_k \boldsymbol{\Theta}^H ( \mathbf{R}_m \mathbf{A}^2 \mathbf{R}_m + \mathrm{tr}\{ \mathbf{A}^2 \mathbf{R}_m\}\mathbf{R}_m))  \} \}
\end{align}
\end{proof}

\begin{cor} \label{col1}
Assuming
\begin{align}
    o_{m, k} &= \hat{q}_{m, k}^* q_{m, k} - \mathbb{E}\{ \hat{q}_{m, k}^* q_{m, k}\}, \\
    o_{m^{\myprime}, k} &= \hat{q}_{m^{\myprime}, k}^* q_{m^{\myprime}, k} - \mathbb{E}\{ \hat{q}_{m, k}^* q_{m, k}\},
\end{align}
where $q_{m, k} = g_{m, k} + \mathbf{h}_m^H \boldsymbol{\Theta} \mathbf{z}_k$ and $\hat{q}_{m, k}$ is given in Lemma \ref{ch_est}. Then, the following holds
\begin{align}
    \mathbb{E}\{o_{m, k} o_{m^{\myprime}, k}^*\} = c_{m, k} c_{m^{\myprime}, k} \sum_{k^{\myprime} \in \mathcal{P}_k }  \mathrm{tr}\{ \boldsymbol{\Xi}_{m, k} \boldsymbol{\Xi}_{m^{\myprime}, k^{\myprime}} \}
\end{align}
\end{cor}
\begin{proof}
The proof is similar to the passive RIS scenario in \cite[Corollary 1]{chien2020massive}.
\end{proof}

\section*{Appendix B} \label{appeB}

The effective uplink signal-to-noise-plus-interference ratio ($\mathrm{SINR}$) for user $k$ is given in (\ref{sinr}) on the top of next page where
   \begin{figure*}[t]
    \begin{equation}
\begin{aligned}
   \Gamma_{k} = \frac{ | E_{k}^{\mathrm{DS}} |^2 }{\mathbb{E}\{| E_{k}^{\mathrm{BU}} 
 |^2 \} + \sum_{k^{\myprime} \neq k}^{K} \mathbb{E}\{| E_{k, k^{\myprime}}^{\mathrm{UI}} 
 |^2 \}   + \mathbb{E}\{| E_{k}^{\mathrm{AN}} 
 |^2 \} + \mathbb{E}\{| E_{k}^{\mathrm{NO}} 
 |^2 \} },
    \label{sinr}
\end{aligned}
    \end{equation}
    \end{figure*}
where
\begin{align} \label{E_DS} \displaybreak[1]
    E_{k}^{\mathrm{DS}} &= \sqrt{\rho_u} \mathbb{E}\left\{ \sum_{m = 1}^{M} \hat{q}_{m, k} q_{m, k} \right\}, \\ \displaybreak[1]
    E_{k}^{\mathrm{BU}} &= \sqrt{\rho_u} \left(  \sum_{m = 1}^{M} \hat{q}_{m, k}^* q_{m, k} -  \mathbb{E}\left\{ \sum_{m = 1}^{M} \hat{q}_{m, k}^* q_{m, k} \right\} \right), \\ \displaybreak[1]
E_{k, k^{\myprime}}^{\mathrm{UI}} &= \sqrt{\rho_u}  \sum_{m = 1}^{M} \hat{q}_{m, k}^* q_{m, k^{\myprime}} , \\
    E_{k}^{\mathrm{AN}} &= \sum_{m = 1}^{M} \hat{q}_{m, k}^* p_m, \\ \displaybreak[1]
    E_{k}^{\mathrm{NO}} &= \sum_{m = 1}^{M} \hat{q}_{m, k}^* w_{m}
\end{align}

Next, (\ref{E_DS}) is calculated as follows
\begin{align} \nonumber
     \left| E_{k}^{\mathrm{DS}} \right|^2 &=  \left| \sqrt{\rho_u} \mathbb{E}\left\{ \sum_{m = 1}^{M} \hat{q}_{m, k} q_{m, k} \right\}\right|^2  \\ \label{DS}
     &= \rho_u \left| \sum_{m=1}^M \mathbb{E}\{ |\hat{q}_{m, k}|^2\} \right|
^2 = \rho_u \left( \sum_{m=1}^M \gamma_{m, k}\right)^2
\end{align}
Next, $     \mathbb{E}\{| E_{k}^{\mathrm{BU}} |^2 \} $ is calculated in (\ref{BU_u}) on the top of the next page.
   \begin{figure*}[t]
    \begin{equation}
\begin{aligned}
     \mathbb{E}\{| E_{k}^{\mathrm{BU}} |^2 \} &= \mathbb{E}\left\{ \left| \sqrt{\rho_u} \left(  \sum_{m = 1}^{M} \hat{q}_{m, k}^* q_{m, k} -  \mathbb{E}\left\{ \sum_{m = 1}^{M} \hat{q}_{m, k}^* q_{m, k}  \right\}  \right)\right|^2 \right\} = \rho_u \mathbb{E} \left\{  \left| \sum_{m = 1}^{M} o_{m, k}   \right|^2 \right\} \\ 
    &= \underbrace{\rho_u \sum_{m = 1}^{M} \sum_{m^{\myprime} = 1, m^{\myprime} \neq m }^{M} \mathbb{E} \left\{ o_{m, k} o_{m^{\myprime}, k}^*  \right\}}_{T_{u0}} + \underbrace{\rho_u \sum_{m = 1}^{M}\mathbb{E} \left\{ |o_{m, k}|^2  \right\} }_{T_{u1}}
    \label{BU_u}
\end{aligned}
    \end{equation}
    \end{figure*}
The close-form expression of $T_{u0}$ based on Corollary~\ref{col1} is as follows
\begin{align}
    T_{u0} = \sum_{m = 1}^{M} \sum_{m^{\myprime} = 1, m^{\myprime} \neq m }^{M}  \sum_{k^{\myprime} \in \mathcal{P}_k} \rho_u c_{m, k} c_{m^{\myprime}, k}  \mathrm{tr}\{ \boldsymbol{\Xi}_{m, k} \boldsymbol{\Xi}_{m^{\myprime}, k^{\myprime}} \}.
\end{align}
Next, the calculation of $T_{u1}$ goes as follows
\begin{align} \nonumber
    T_{u1} &= \rho_u \sum_{m = 1}^{M}\mathbb{E} \left\{ |o_{m, k}|^2  \right\} \\ \nonumber
    &= \rho_u \sum_{m = 1}^{M} \mathbb{E}\left\{ \left| \hat{q}_{m, k}^* q_{m, k} -  \mathbb{E}\{ \hat{q}_{m, k}^* q_{m, k} \} \right|^2 \right\} \\ \nonumber
    &= \rho_u \sum_{m = 1}^{M} \mathbb{E}\{ |\hat{q}_{m, k}^* q_{m, k}|^2 \} - \rho_u \sum_{m = 1}^{M} \left| \mathbb{E}\{ \hat{q}_{m, k}^* q_{m, k}\} \right|^2 \\ 
    &= \rho_u \sum_{m = 1}^{M} \mathbb{E}\{ |\hat{q}_{m, k}^* q_{m, k}|^2 \} - \rho_u \sum_{m = 1}^{M} \gamma_{m, k}^2.
\end{align}
Furthermore, $ \mathbb{E}\{ |\hat{q}_{m, k}^* q_{m, k}|^2 \}$ could be calculated in (\ref{E_qq2}) at the top of next page.
   \begin{figure*}[t]
    \begin{equation}
\begin{aligned}
    &\mathbb{E}\{ |\hat{q}_{m, k}^* q_{m, k}|^2 \} =  \mathbb{E}\{ |c_{m, k} y_{m, k}^{(\mathrm{p})*} q_{m, k}|^2 \} = c_{m, k}^2 \mathbb{E}\left\{ \left|\left((\sum_{k^{\myprime} \in \mathcal{P}_k}  q_{m, k^{\myprime}}^*  + \frac{1}{\sqrt{\rho \tau_p}} \Bar{p}_{m, k}^{(\mathrm{p})*} + \frac{1}{\sqrt{\rho \tau_p}} w_{m, k}^{(\mathrm{p})*} \right) q_{m, k}\right|^2 \right\} \\ 
     &= c_{m, k}^2  \mathbb{E}\{ |q_{m, k}|^4 \} + c_{m, k}^2 \sum_{k^{\myprime} \in \mathcal{P}_k \backslash \{k\}}   \mathbb{E}\{ |q_{m, k^{\myprime}}^* q_{m, k}|^2 \} + \frac{c_{m, k}^2}{\rho \tau_p} \mathbb{E}\{|w_{m, k}^{(\mathrm{p})*}q_{m, k}|^2 \} +  \frac{c_{m, k}^2}{\rho \tau_p} \mathbb{E}\{|\Bar{p}_{m, k}^{(\mathrm{p})*}q_{m, k}|^2 \}.
     \label{E_qq2}
\end{aligned}
    \end{equation}
    \end{figure*}
In what follows, each term of (\ref{E_qq2}) is calculated. First, the expression of $\mathbb{E}\{ |q_{m, k}|^4 \}$ is given in (\ref{e_qq4}). The second term is also calculated
in (\ref{E_qmkqmkk}). Also, $\mathbb{E}\{|w_{m, k}^{(\mathrm{p})*}q_{m, k}|^2 \}$ thanks to the independency of the noise and the channel could be written as $\mathbb{E}\{|w_{m, k}^{(\mathrm{p})*}q_{m, k}|^2 \} = \sigma^2 \kappa_{m, k}$. Finally, $\mathbb{E}\{|\Bar{p}_{m, k}^{(\mathrm{p})*}q_{m, k}|^2 \}$ is given in (\ref{e_bq}). Finally, we can write (\ref{BU_u}) as follows
\begin{align}\nonumber
    \mathbb{E}\{| E_{k}^{\mathrm{BU}} |^2 \} &=  \sum_{m = 1}^{M} \sum_{m^{\myprime} = 1 }^{M}  \sum_{k^{\myprime} \in \mathcal{P}_k} \rho_u c_{m, k} c_{m^{\myprime}, k}  \mathrm{tr}\{ \boldsymbol{\Xi}_{m, k} \boldsymbol{\Xi}_{m^{\myprime}, k^{\myprime}} \} \\ \nonumber
    &+ \sum_{m = 1}^{M}  \rho_u \gamma_{m, k}^2 + \sum_{m = 1}^{M} 2 \rho_uc_{m, k}^2\mathrm{tr}\{ \boldsymbol{\Xi}_{m, k}^2 \} \\ \nonumber
    &+ \sum_{m = 1}^{M} \sum_{k^{\myprime} \in \mathcal{P}_k \backslash \{k\}} \rho_u c_{m, k}^2 \kappa_{m, k}\kappa_{m, k^{\myprime}} \\ \label{BUu}
    &+  \sum_{m = 1}^{M} \frac{\rho_u c_{m, k}^2 \sigma^2 \kappa_{m, k}}{\rho \tau_p} + \sum_{m = 1}^{M} \frac{\rho_u c_{m, k}^2 \alpha_{m, k}}{\rho \tau_p}. 
\end{align}
The next term of the denominator of (\ref{sinr}) could be written as follows
\begin{align} \nonumber
     \sum_{k^{\myprime} = 1, k^{\myprime} \neq k}^{K} \mathbb{E}\{| E_{k, k^{\myprime}}^{\mathrm{UI}} 
 |^2 \} &= \sum_{k^{\myprime} \notin \mathcal{P}_k} \mathbb{E}\{| E_{k, k^{\myprime}}^{\mathrm{UI}} |^2 \} \\ \label{inter_1}
 &+  \sum_{k^{\myprime} \in \mathcal{P}_k \backslash \{k\}} \mathbb{E}\{| E_{k, k^{\myprime}}^{\mathrm{UI}} |^2 \}.
\end{align}
The first term of (\ref{inter_1})  could be simplified as follows
\begin{align} \nonumber
     &\sum_{k^{\myprime} \notin \mathcal{P}_k} \mathbb{E}\{| E_{k, k^{\myprime}}^{\mathrm{UI}} |^2 \} = \sum_{k^{\myprime} \notin \mathcal{P}_k} \mathbb{E}\{| \sqrt{\rho_u}  \sum_{m = 1}^{M} \hat{q}_{m, k}^* q_{m, k^{\myprime}} |^2 \} \\ \label{E_out}
     &= \rho_u \sum_{k^{\myprime} \notin \mathcal{P}_k} \sum_{m = 1}^{M} \sum_{m^{\myprime} = 1}^{M}  \mathbb{E}\{  \hat{q}_{m, k}^* q_{m, k^{\myprime}}  \hat{q}_{m^{\myprime}, k} q_{m^{\myprime}, k^{\myprime}}^* \}.
\end{align}
The expectation term in (\ref{E_out}) can be written as follows
\begin{align} \nonumber
    \mathbb{E}\{  \hat{q}_{m, k}^*  & q_{m, k^{\myprime}}  \hat{q}_{m^{\myprime}, k} q_{m^{\myprime}, k^{\myprime}}^* \} \\ \nonumber
    &= c_{m, k} c_{m^{\myprime}, k} \sum_{k^{\mydprime} \in \mathcal{P}_k} \mathbb{E}\{  q_{m, k^{\mydprime}}^* q_{m, k^{\myprime}} q_{m^{\myprime}, k^{\mydprime}} q_{m^{\myprime}, k^{\myprime}}^*\} \\ \nonumber
    &+  \frac{c_{m, k} c_{m^{\myprime}, k}}{\rho \tau_p} \mathbb{E} \{  \Bar{p}_{m, k}^{(\mathrm{p})*}  q_{m, k^{\myprime}}  \Bar{p}_{m^{\myprime}, k}^{(\mathrm{p})} q_{m^{\myprime}, k^{\myprime}}^*\} \\ \nonumber
    &+  \frac{c_{m, k} c_{m^{\myprime}, k}}{\rho \tau_p}  \mathbb{E} \{  w_{m, k}^{(\mathrm{p})*}  q_{m, k^{\myprime}}  w_{m^{\myprime}, k}^{(\mathrm{p})} q_{m^{\myprime}, k^{\myprime}}^*\} \\ \label{E_inner}
    &= \begin{cases}
      c_{m, k} c_{m^{\myprime}, k} \sum_{k^{\mydprime} \in \mathcal{P}_k} T_{m, k^{^{\myprime}}, m^{\myprime}, k^{\mydprime}} & \text{, if $m \neq m^{\myprime}$}\\
      T_{m, k^{\myprime}, m, k} & \text{, if $m = m^{\myprime}$}
    \end{cases} .
\end{align}
where 
\begin{align}
&T_{m,k^{^{\myprime}} ,m^{\myprime},k^{\mydprime}} = \mathbb{E}\{  q_{m, k^{\mydprime}}^* q_{m, k^{\myprime}} q_{m^{\myprime}, k^{\mydprime}} q_{m^{\myprime}, k^{\myprime}}^*\} \\ \nonumber
      & T_{m, k^{\myprime}, m, k} = c_{m, k}^2 \sum_{k^{\mydprime} \in \mathcal{P}_k} \mathbb{E}\{  |q_{m, k^{\mydprime}}^* q_{m, k^{\myprime}}|^2 \} \\
      &+  \frac{c_{m, k}^2}{\rho \tau_p} \mathbb{E} \{  |\Bar{p}_{m, k}^{(\mathrm{p})*}  q_{m, k^{\myprime}} |^2 \} + \frac{c_{m, k}^2}{\rho \tau_p} \mathbb{E} \{  |w_{m, k}^{(\mathrm{p})*}  q_{m, k^{\myprime}}|^2 \} .
\end{align}
The expression for $T_{m,k^{^{\myprime}} ,m^{\myprime},k^{\mydprime}}$ is calculated in (\ref{eqqqq}) and  is follows 
\begin{align}\label{T_mkmkmkm}
    T_{m,k^{^{\myprime}} ,m^{\myprime},k^{\mydprime}} = \mathrm{tr}\{ \boldsymbol{\Xi}_{m, k^{\myprime}} \boldsymbol{\Xi}_{m^{\myprime}, k^{\mydprime}} \} .
\end{align}
Moreover,  $T_{m, k^{\myprime}, m, k}$ could be calculated  as follows
\begin{align}\nonumber
    T_{m, k^{\myprime}, m, k} &= c_{m, k}^2 \sum_{k^{\mydprime} \in \mathcal{P}_k} \kappa_{m, k^{\mydprime}}\kappa_{m, k^{\myprime}} \\ \nonumber
    &+ c_{m, k}^2 \sum_{k^{\mydprime} \in \mathcal{P}_k} \mathrm{tr}\{ \boldsymbol{\Xi}_{m, k^{\mydprime}} \boldsymbol{\Xi}_{m, k^{\myprime}} \} \\ \label{T_mkkmk}
    &+ \frac{c_{m, k}^2  \alpha_{m, k^{\myprime}}}{\rho \tau_p} + \frac{c_{m, k}^2 \sigma^2 \kappa_{m, k^{\myprime}}}{\rho \tau_p} .
\end{align}
Based on (\ref{T_mkmkmkm}) and (\ref{T_mkkmk}) , the first term of the right-hand-side (RHS) of (\ref{inter_1}) is given as follows
\begin{align} \nonumber
    &\sum_{k^{\myprime} \in \mathcal{P}_k} \mathbb{E}\{| E_{k, k^{\myprime}}^{\mathrm{UI}} |^2 \} = \rho_u  \sum_{k^{\myprime} \notin \mathcal{P}_k}  
      \sum_{k^{\mydprime} \in \mathcal{P}_k} \sum_{m = 1}^{M} c_{m, k}^2 \kappa_{m, k^{\mydprime}}\kappa_{m, k^{\myprime}} \\ \nonumber
      &+  \rho_u \sum_{k^{\myprime} \notin \mathcal{P}_k}   \sum_{m = 1}^{M} \frac{c_{m, k}^2  \alpha_{m, k^{\myprime}}}{\rho \tau_p} \\ \nonumber
      &+  \rho_u \sum_{k^{\myprime} \notin \mathcal{P}_k}   \sum_{m = 1}^{M} \frac{c_{m, k}^2 \sigma^2 \kappa_{m, k^{\myprime}}}{\rho \tau_p}  \\  \label{rhs_1}
      &+  \rho_u \sum_{k^{\myprime} \notin \mathcal{P}_k} \sum_{k^{\mydprime} \in \mathcal{P}_k}  \sum_{m = 1}^{M} \sum_{m^{\myprime} = 1}^{M}  c_{m, k} c_{m^{\myprime}, k} \mathrm{tr}\{ \boldsymbol{\Xi}_{m, k^{\myprime}} \boldsymbol{\Xi}_{m^{\myprime}, k^{\mydprime}} \} . 
\end{align}
The second term of RHS of the (\ref{inter_1}) could be written as (\ref{snd_eqq}) at the top of the next page.
\begin{figure*}[t]
\begin{equation}
\begin{aligned}
    \mathbb{E}\{| E_{k, k^{\myprime}}^{\mathrm{UI}} |^2 \} &=  \mathbb{E}\{| \sqrt{\rho_u}  \sum_{m = 1}^{M} \hat{q}_{m, k}^* q_{m, k^{\myprime}} |^2 \} = \rho_u \mathbb{E}\left\{  \left|  \sum_{m = 1}^{M} c_{m, k} y_{m, k}^{(\mathrm{p})*} q_{m, k^{\myprime}} \right|^2 \right\} \\  
    &= \rho_u \mathbb{E}\left\{  \left|  \sum_{m = 1}^{M} c_{m, k} \left( \sum_{k^{\mydprime} \in \mathcal{P}_k }  q_{m, k^{\mydprime}}^*  + \frac{1}{\sqrt{\rho \tau_p}} \Bar{p}_{m, k}^{(\mathrm{p})*} + \frac{1}{\sqrt{\rho \tau_p}} w_{m, k}^{(\mathrm{p})*} \right) q_{m, k^{\myprime}} \right|^2 \right\} \\ 
    &= \rho_u  \mathbb{E}\left\{  \left|c_{m, k} \sum_{m = 1}^{M} |q_{m, k^{\myprime}}|^2 \right|^2  \right\} + \rho_u \mathbb{E}\left\{ \left|  \sum_{m = 1}^{M} c_{m, k} \left(   \sum_{k^{\mydprime} \in \mathcal{P}_k \backslash \{ k^{\myprime} \}} q_{m, k^{\mydprime}}^* \right)  q_{m, k^{\myprime}}\right|^2  \right\} \\  
    &+ \frac{\rho_u}{\rho \tau_p} \mathbb{E}\left\{ \left|  \sum_{m = 1}^{M} c_{m, k}  \Bar{p}_{m, k}^{(\mathrm{p})*} q_{m, k^{\myprime}} \right|^2 \right\} + \frac{\rho_u}{\rho \tau_p} \mathbb{E}\left\{ \left|  \sum_{m = 1}^{M} c_{m, k}  w_{m, k}^{(\mathrm{p})*} q_{m, k^{\myprime}} \right|^2 \right\}.
    \label{snd_eqq} 
\end{aligned}
    \end{equation}
\end{figure*}

In what follows, each term of (\ref{snd_eqq}) will be calculated separately. The first term of RHS of  (\ref{snd_eqq}) could be written as
\begin{align} \nonumber
    \rho_u & \mathbb{E}\left\{  \left|c_{m, k} \sum_{m = 1}^{M} |q_{m, k^{\myprime}}|^2  \right|^2  \right\} \\ \displaybreak[0] \nonumber
    &=  \rho_u  \sum_{m = 1}^{M} \sum_{m^{\myprime} = 1}^{M}   c_{m, k} c_{m^{\myprime}, k} \mathbb{E}\left\{ |q_{m, k^{\myprime}}|^2 |q_{m^{\myprime}, k^{\myprime}}|^2   \right\} \\\nonumber \displaybreak[0]
    &=  \rho_u  \sum_{m = 1}^{M} c_{m, k}^2 \mathbb{E}\left\{ |q_{m, k^{\myprime}}|^4   \right\} \\ \displaybreak[0] \nonumber
    &+ \rho_u  \sum_{m = 1}^{M} \sum_{m^{\myprime} = 1, m^{\myprime} \neq m}^{M}   c_{m, k} c_{m^{\myprime}, k} \mathbb{E}\left\{ |q_{m, k^{\myprime}}|^2 |q_{m^{\myprime}, k^{\myprime}}|^2   \right\} \\ \displaybreak[0] \nonumber 
    &= 2\rho_u  \sum_{m = 1}^{M}  c_{m, k}^2  \kappa_{m, k^{\myprime}}^2 + 2\rho_u  \sum_{m = 1}^{M}  c_{m, k}^2  \mathrm{tr}\{ \boldsymbol{\Xi}_{m, k^{\myprime}}^2 \} \\ \displaybreak[0] \nonumber
    &+  \rho_u  \sum_{m = 1}^{M} \sum_{m^{\myprime} = 1, m^{\myprime} \neq m}^{M}   c_{m, k} c_{m^{\myprime}, k} \kappa_{m, k^{\myprime}}\kappa_{m^{\myprime}, k^{\myprime}} \\ \displaybreak[0] \nonumber 
    &+  \rho_u  \sum_{m = 1}^{M} \sum_{m^{\myprime} = 1, m^{\myprime} \neq m}^{M}   c_{m, k} c_{m^{\myprime}, k} \mathrm{tr}\{ \boldsymbol{\Xi}_{m^{\myprime}, k^{\myprime}} \boldsymbol{\Xi}_{m, k^{\myprime}} \} \\ \displaybreak[0] \nonumber 
    &= \rho_u  \sum_{m = 1}^{M}  c_{m, k}^2  \kappa_{m, k^{\myprime}}^2 + \rho_u  \sum_{m = 1}^{M}  c_{m, k}^2  \mathrm{tr}\{ \boldsymbol{\Xi}_{m, k^{\myprime}}^2 \} \\ \displaybreak[0] \nonumber
    &+ \rho_u  \left( \sum_{m = 1}^{M}   c_{m, k}  \kappa_{m, k^{\myprime}} \right)^2 \\ \displaybreak[0] \label{rhs_1_eqq} 
    &+  \rho_u  \sum_{m = 1}^{M} \sum_{m^{\myprime} = 1}^{M}   c_{m, k} c_{m^{\myprime}, k} \mathrm{tr}\{ \boldsymbol{\Xi}_{m^{\myprime}, k^{\myprime}} \boldsymbol{\Xi}_{m, k^{\myprime}} \}.
\end{align}
The second term of RHS of  (\ref{snd_eqq}) could be written as (\ref{rhs_2_eqq}) at the top of the next page.

   \begin{figure*}[t]
    \begin{equation}
\begin{aligned}
  &\rho_u \mathbb{E}\left\{ \left|  \sum_{m = 1}^{M} c_{m, k} \left(   \sum_{k^{\mydprime} \in \mathcal{P}_k \backslash \{ k^{\myprime} \}} q_{m, k^{\mydprime}}^* \right)  q_{m, k^{\myprime}}\right|^2  \right\} \\ 
    &=  \rho_u\sum_{k^{\mydprime} \in \mathcal{P}_k \backslash \{ k^{\myprime} \}}  \sum_{m = 1}^{M} c_{m, k}^2  \mathbb{E}\left\{ \left|  q_{m, k^{\mydprime}} \right|^2  \left|  q_{m, k^{\myprime}} \right|^2 \right\} \\
    &+ \rho_u\sum_{k^{\mydprime} \in \mathcal{P}_k \backslash \{ k^{\myprime} \}}  \sum_{m = 1}^{M}  \sum_{m^{\myprime} = 1, m^\prime \neq m}^{M} c_{m, k} c_{m^{\myprime}, k} \mathbb{E}\{ q_{m, k^{\mydprime}}^*  q_{m, k^{\myprime}} q_{m^{\myprime}, k^{\mydprime}} q_{m^{\myprime}, k^{\myprime}}^* \} \\ 
    &= \rho_u\sum_{k^{\mydprime} \in \mathcal{P}_k \backslash \{ k^{\myprime} \}}  \sum_{m = 1}^{M}   c_{m, k}^2 \kappa_{m, k^{\mydprime}}\kappa_{m, k^{\myprime}} + \rho_u\sum_{k^{\mydprime} \in \mathcal{P}_k \backslash \{ k^{\myprime} \}}  \sum_{m = 1}^{M}  \sum_{m^{\myprime} = 1}^{M} c_{m, k} c_{m^{\myprime}, k} \mathrm{tr}\{ \boldsymbol{\Xi}_{m, k^{\myprime}} \boldsymbol{\Xi}_{m^{\myprime}, k^{\mydprime}} \} 
    \label{rhs_2_eqq}
\end{aligned}
    \end{equation}
    \end{figure*}

The third term of RHS of  (\ref{snd_eqq}) is given as follows
\begin{align} \nonumber
    &\frac{\rho_u}{\rho \tau_p} \mathbb{E}\left\{ \left|  \sum_{m = 1}^{M} c_{m, k}  \Bar{p}_{m, k}^{(\mathrm{p})*} q_{m, k^{\myprime}} \right|^2 \right\} \\ 
    &=  \frac{\rho_u}{\rho \tau_p} \sum_{m = 1}^{M} c_{m, k}^2 \mathbb{E}\left\{ \left|  \Bar{p}_{m, k}^{(\mathrm{p})*} q_{m, k^{\myprime}} \right|^2 \right\} = \frac{\rho_u}{\rho \tau_p} \sum_{m = 1}^{M} c_{m, k}^2 \alpha_{m, k^{\myprime}}.
\end{align}
Finally, the fourth term of RHS of  (\ref{snd_eqq}) is
\begin{align}
    \frac{\rho_u}{\rho \tau_p} \mathbb{E}\left\{ \left|  \sum_{m = 1}^{M} c_{m, k}  w_{m, k}^{(\mathrm{p})*} q_{m, k^{\myprime}} \right|^2 \right\} = \frac{\rho_u \sigma^2}{\rho \tau_p} \sum_{m = 1}^{M} c_{m, k}^2 \kappa_{m, k^{\myprime}}.
\end{align}

So far, all the terms of (\ref{snd_eqq})
 is calculated and is given as follows
 \begin{align} \label{rhs_2}
      &\sum_{k^{\myprime} \in \mathcal{P}_k \backslash \{k\}} \mathbb{E}\{| E_{k, k^{\myprime}}^{\mathrm{UI}} |^2 \} = \rho_u \sum_{k^{\myprime} \in \mathcal{P}_k \backslash \{k\}} \sum_{m = 1}^{M}  c_{m, k}^2  \kappa_{m, k^{\myprime}}^2 \\ \nonumber
      &+ \rho_u \sum_{k^{\myprime} \in \mathcal{P}_k \backslash \{k\}}  \sum_{m = 1}^{M}  c_{m, k}^2  \mathrm{tr}\{ \boldsymbol{\Xi}_{m, k^{\myprime}}^2 \} \\ \nonumber
      &+ \rho_u \sum_{k^{\myprime} \in \mathcal{P}_k \backslash \{k\}}  \left( \sum_{m = 1}^{M}   c_{m, k}  \kappa_{m, k^{\myprime}} \right)^2 \\ \nonumber
      &+ \rho_u \sum_{k^{\myprime} \in \mathcal{P}_k \backslash \{k\}} \sum_{k^{\mydprime} \in \mathcal{P}_k \backslash \{ k^{\myprime} \}}  \sum_{m = 1}^{M}   c_{m, k}^2 \kappa_{m, k^{\mydprime}}\kappa_{m, k^{\myprime}} \\ \nonumber
    &+ \rho_u \sum_{k^{\myprime} \in \mathcal{P}_k \backslash \{k\}} \sum_{k^{\mydprime} \in \mathcal{P}_k }  \sum_{m = 1}^{M}  \sum_{m^{\myprime} = 1}^{M} c_{m, k} c_{m^{\myprime}, k} \mathrm{tr}\{ \boldsymbol{\Xi}_{m, k^{\myprime}} \boldsymbol{\Xi}_{m^{\myprime}, k^{\mydprime}} \} \\ \nonumber
    &+ \frac{\rho_u}{\rho \tau_p} \sum_{k^{\myprime} \in \mathcal{P}_k \backslash \{k\}} \sum_{m = 1}^{M} c_{m, k}^2 \alpha_{m, k^{\myprime}} \\ \nonumber
    &+ \frac{\rho_u \sigma^2}{\rho \tau_p} \sum_{k^{\myprime} \in \mathcal{P}_k \backslash \{k\}} \sum_{m = 1}^{M} c_{m, k}^2 \kappa_{m, k^{\myprime}}.
 \end{align}
Based on (\ref{rhs_1}) and (\ref{rhs_2}), (\ref{inter_1}) could be written as (\ref{UI}) at the top of the next page.

   \begin{figure*}[t]
    \begin{equation}
\begin{aligned}
&\sum_{k^{\myprime} = 1, k^{\myprime} \neq k}^{K} \mathbb{E}\{| E_{k, k^{\myprime}}^{\mathrm{UI}} 
 |^2 \} =  \rho_u \sum_{k^{\myprime}=1, k^{\myprime} \neq k}^K \sum_{k^{\mydprime} \in \mathcal{P}_k}  \sum_{m = 1}^{M} \sum_{m^{\myprime} = 1}^{M}  c_{m, k} c_{m^{\myprime}, k} \mathrm{tr}\{ \boldsymbol{\Xi}_{m, k^{\myprime}} \boldsymbol{\Xi}_{m^{\myprime}, k^{\mydprime}} \} \\ 
 & + \rho_u  \sum_{k^{\myprime} =1, k^{\myprime} \neq k}^K  
      \sum_{k^{\mydprime} \in \mathcal{P}_k} \sum_{m = 1}^{M} c_{m, k}^2 \kappa_{m, k^{\mydprime}}\kappa_{m, k^{\myprime}}+  \frac{\rho_u}{\rho \tau_p} \sum_{k^{\myprime} =1, k^{\myprime}\neq k}^K   \sum_{m = 1}^{M}c_{m, k}^2  \alpha_{m, k^{\myprime}} \\ 
      &+  \frac{\rho_u \sigma^2}{\rho \tau_p} \sum_{k^{\myprime} =1, k^{\myprime} \neq k}   \sum_{m = 1}^{M} c_{m, k}^2 \kappa_{m, k^{\myprime}} + \rho_u \sum_{k^{\myprime} \in \mathcal{P}_k \backslash \{k\}}  \left( \sum_{m = 1}^{M}   c_{m, k}  \kappa_{m, k^{\myprime}} \right)^2   \\ 
      &+ \rho_u \sum_{k^{\myprime} \in \mathcal{P}_k \backslash \{k\}} \sum_{m = 1}^{M}  c_{m, k}^2  \kappa_{m, k^{\myprime}}^2 + \rho_u \sum_{k^{\myprime} \in \mathcal{P}_k \backslash \{k\}}  \sum_{m = 1}^{M}  c_{m, k}^2  \mathrm{tr}\{ \boldsymbol{\Xi}_{m, k^{\myprime}}^2 \} 
    \label{UI}
\end{aligned}
    \end{equation}
    \end{figure*}

The next term to be calculated is $E_{k}^{\mathrm{AN}}$ and is given as follows
\begin{align} \label{AN}
    \mathbb{E}\{|E_{k}^{\mathrm{AN}}|^2\} &= \sum_{m=1}^M \mathbb{E}\{|\hat{q}_{m, k}^* p_m|^2 \} = \sum_{m=1}^M \alpha_{m, k}
\end{align}

The noise term of the denominator of (\ref{sinr}) can be written as follows
\begin{align} \label{NO}
    \mathbb{E}\{| E_{k}^{\mathrm{NO}} 
 |^2 \} = \sum_{m = 1}^{M}    \mathbb{E}\{|\hat{q}_{m, k}^* w_{m} |^2 \}  = \sigma^2 \sum_{m = 1}^{M} \kappa_{m, k}.
\end{align}

Finally, by inserting (\ref{DS}), (\ref{BUu}) (\ref{UI}), (\ref{AN}) and (\ref{NO}) in (\ref{sinr}) and with the aid of some algebraic manipulations, we arrive at (\ref{sinr_simple}) which completes the proof.

\section*{Appendix C} \label{appC}
The total power consumption is modeled as \cite{ngo2017total, bjornson2015optimal}
\begin{equation} \label{total_p}
    P_{\mathrm{total}} = K \zeta  \rho_u + P_{\mathrm{bh}, m} + P_{\mathrm{ARIS}}
\end{equation}
where $0<\zeta \leq 1$ is the power amplifier efficiency of each user. Also, $P_{\mathrm{ARIS}}$  is given in  (\ref{p_aris}). Furthermore, $P_{\mathrm{bh}, m}$ is the backhaul that transfers the data between the CPU and the APs. More specifically, its power consumption can be written as
\begin{equation}
    P_{\mathrm{bh}, m} = P_{0, m} + B \times\mathrm{SE}_k \times P_{\mathrm{bt}, m},
\end{equation}
where $P_{0, m}$ is the fixed power consumption of each backhaul and $P_{\mathrm{bt}, m}$ is the traffic-dependent power (in $\text{Watt per bit/sec}$) and $B$ is the system bandwidth.

Finally, based on the total power consumption in (\ref{total_p}), the EE is defined  as
\begin{equation} \label{total_ee}
   \iota =  \frac{B \sum_{k=1}^K R_k}{P_{\mathrm{total}}}
\end{equation}

\end{document}